\documentclass[a4paper,UKenglish,cleveref, autoref, thm-restate]{lipics-v2021}

\title{Reachability in $3$-VASS is in \textsc{Tower}} 

\author{Qizhe Yang}{Shanghai Normal University}{qzyang@shnu.edu.cn}{}
{}
\author{Yuxi Fu}{BASICS, Shanghai Jiaotong University}{fu-yx@cs.sjtu.edu.cn}{}{}

\authorrunning{Qizhe Yang, Yuxi Fu} 

\Copyright{Qizhe Yang, Yuxi Fu} 



\acknowledgements{We thank the members of BASICS for their informative discussions and criticisms.
We thank Yangluo Zheng for helpful comments and for pointing out a critical mistake in an early draft.}
\nolinenumbers 

\usepackage{nccmath}
\usepackage[inference]{semantic}
\usepackage{mathdots}
\usepackage{mathrsfs}
\usepackage{tikz}
\usepackage{soul}
\usetikzlibrary {fit,shapes.geometric}
\usetikzlibrary {arrows.meta}
\usetikzlibrary {quotes}
\usetikzlibrary {positioning}
\usetikzlibrary {calc}
\input xy
\xyoption{all}

\allowdisplaybreaks[3]

\newcommand{\defn}{\stackrel{\text{def}}{=}}

\begin{document}

\maketitle

\begin{abstract}
The reachability problem for vector addition systems with states (VASS) is \textsc{Ackermann}-complete.
For every $k\geq 3$, a completeness result for the $k$-dimensional VASS reachability problem is not yet available.
It is shown in this paper that the $3$-dimensional VASS reachability problem is in \textsc{Tower}, improving upon  the current best upper bound $\mathbf{F}_7$ established by Leroux and Schmidt in 2019.
\end{abstract}

\section{Introduction}\label{Sec-Introduction}

Petri net theory has been studied for over half a century.
As a model for concurrency and causality, Petri net models find a wide range of applications in system specification and verification~\cite{carvalho2020mobile,heiner2008petri}.
Two equivalent formulations of Petri nets, vector addition systems $(\text{VAS})$ and vector addition systems with states $(\text{VASS})$ have been studied for a long time.
In vector addition systems configurations are formulated as vectors on non-negative integers. Computation rules of systems are captured by vectors on integers. A computation is a sequence of legal transitions of configurations. The reachability problem asks whether a target configuration is reachable from an initial configuration in a given VAS. The problem holds a central position in the study of VAS, as numerous issues in the areas of language, logic, and concurrency can be effectively reduced to this particular problem~\cite{schmitz2016complexity}.

The decidability result of the VASS reachability problem is among the most significant theoretical breakthroughs in computer science.
In the 1970s, some decidability results~\cite{hopcroft1979reachability,leeuwen1974partial} for low dimensional VASSes have been established.
The initial work of Sacerdote and Tenney~\cite{sacerdote1977decidability} gives an incomplete proof of the decidability of the problem.
A complete decidability proof of the problem was given in the early 1980s by Mayr~\cite{mayr1981complexity}.
The involved proof was later refined by Kosaraju~\cite{kosaraju1982decidability} and Lambert~\cite{lambert1992structure}.
The algorithm is now referred to as KLMST decomposition.
In recent years, Leroux gives another proof of the decidability in a more logic setting using Presburger invariants~\cite{leroux2012vector}.

In 2015 Leroux and Schmitz obtained the first upper bound for the KLMST algorithm, pointing out that it is in the cubic-Ackermann complexity class $\mathbf{F}_{\omega^{3}}$~\cite{leroux2015demystifying}.
The upper bound was improved to $\mathbf{F}_{\omega}$~\cite{leroux2019reachability} by Leroux and Schmitz in 2019.
The \textbf{EXPSPACE}-hardness of the problem was shown by Lipton in 1976~\cite{lipton1976reachability}.
For many years this has been the only lower bound we knew.
Until in 2018 Czerwi{\'n}ski et al came up with the ``Amplifier'' technique and applied it to obtain a non-elementary lower bound~\cite{czerwinski2020reachability}.
This was improved to $\mathbf{F}_{\omega}$ in 2022, independently by several groups~\cite{czerwinski2021reachability,lasota2021improved,leroux2021reachability}. The Ackermann-completeness of the problem is thus established.

Reachability of VASSes with fixed dimension has also gathered widespread attention.
Completeness results have been established for low dimensional VASS.
Haase {\em et al} showed that the reachability of $1$-dimensional VASS is NP-complete~\cite{haase2009reachability}.
In the $2$-dimensional case, Blondin {\em et al} proved that the problem is PSPACE complete~\cite{blondin2015reachability} and Englert, Lazi{\'c} and Totzke pointed out that the problem is NL-complete~\cite{englert2016reachability} if unary encoding is used.
It has been known from the early stage that while reachability sets are semi-linear in the two dimensional case, they are not semi-linear for VASSes in three or more dimensions~\cite{hopcroft1979reachability}.
This stops us from generalizing the proof of the completeness result for $2$-VASS to higher dimensional VASS.
A completeness characterization of the $k$-VASS reachability problem, where $k\ge3$, remains open.
Currently it is known that the reachability problem for $d$-VASS with $d\geq 3$ is in $\mathbf{F}_{d+4}$~\cite{leroux2019reachability}.
For the lower bound Leroux proved in~\cite{leroux2021reachability} that reachability in $(2d{+}4)$-VASS is $\mathbf{F}_d$-hard for $d\ge3$, and Czerwi{\'n}ski proved that reachability in $8$-VASS is non-elementary~\cite{czerwinski2020fixreachability,czerwinski2022lower}. For $3$-VASS in particular, the best lower bound is \textbf{PSPACE}-hard~\cite{blondin2015reachability} and the best upper bound is $\mathbf{F}_7$~\cite{leroux2019reachability}.
There remains a substantial gap between the known lower bounds and the known upper bounds for the reachability problem of VASSes with fixed dimension.

The main contribution of the present paper is stated in the following theorem.
\begin{theorem}\label{MAIN-THEOREM}
Reachability in $3$-VASSes is in \textsc{Tower}.
\end{theorem}
The theorem is proved by modifying the well-known KLMST decomposition algorithm for the general VASS reachability problem.
Our algorithm incorporates the linear path scheme characterization for $2$-VASS to the general KLMST algorithm.
We show that a kind of special $3$-VASS, to be called effectively $2$-dimensional, has the linear path scheme property.
Based on this observation the algorithm replaces every newly generated effectively $2$-dimensional component immediately by a linear path scheme, preventing further decomposition of the effectively $2$-dimensional component.
In this way the depth of decomposition tree is bounded by a linear function, hence the \textsc{Tower} upper bound.
The techniques developed for proving Theorem~\ref{MAIN-THEOREM} can be applied to derive the following result.
\begin{theorem}\label{cor:reachabilityforeffectively2dim3vass}
Reachability in the effectively $2$-dimensional $3$-VASSes is in $\mathbf{EXPSPACE}$.
\end{theorem}

The rest of the paper is organized as follows.
Section~\ref{Sec-Preliminary} states the preliminaries.
Section~\ref{sectionlinearpathscheme} reviews the linear path schemes for the $2$-VASSes and extends the technique to the effectively $2$-dimensional $3$-VASSes.
Section~\ref{sectionklmstalgo} recalls the main KLMST constructions.
Section~\ref{sectionalmostnormal} defines the $3$-normal KLM sequences and derives a tower space algorithm for the $3$-VASSes, which is the main contribution of the paper.
Section~\ref{sectionconclusion} makes a few comments.

\section{Preliminaries}\label{Sec-Preliminary}

Let $\mathbb{N}$ be the set of natural numbers (nonnegative integers) and $\mathbb{Z}$ the set of integers.
Let $\mathbb{V}$ denote the set of variables for nonnegative integers.
For $L\in\mathbb{N}\setminus\{0\}$ the notation $[L]$ stands for the set $\{1,\ldots,L\}$ and $[L]_{0}$ for $\{0\}\cup[L]$.
For a finite set $S$ let $|S|$ denote the number of elements of $S$.
We introduce a super number $\omega$ with $n<\omega$ for all $n\in\mathbb{N}$.
Intuitively $\omega$ stands for a number that can be as large as necessary.
Let $\mathbb{N}_{\omega}=\mathbb{N}\cup\{\omega\}$ be the extended set of natural numbers.
The partial order $\sqsubseteq$ over $\mathbb{N}_{\omega}$ is defined as follows: $x\sqsubseteq y$ whenever $y\in\{x,\omega\}$.

We write $\mathbf{m},\mathbf{n}$ for $d$-dimensional vectors in $\mathbb{N}^d$, $\mathbf{u},\mathbf{v}$ for vectors in $\mathbb{N}_{\omega}^d$, and $\mathbf{x},\mathbf{y}$ for vectors in $\mathbb{V}^d$.
For $i\in[d]$ we write for example $\mathbf{a}(i)$ for the {\em $i$-th entry} of $\mathbf{a}$.
Let $\mathbf{1}=(1,\ldots,1)^{\dag}$ and $\mathbf{0}=(0,\ldots,0)^{\dag}$, where $(\_)^{\dag}$ is the transposition operator.
We write $\sigma$ for a finite sequence of vectors and $|\sigma|$ for the length of the sequence.
For $i\in[|\sigma|]$ we write $\sigma[i]$ for the $i$-th element that appears in $\sigma$.
The notation  $\sigma[i,\ldots,j]$ is for the subsequence $\sigma[i]\sigma[i+1]\ldots\sigma[j]$ if $i\leq j$ and is for $\epsilon$ if $i>j$.

Recall that the {\em $1$-norm} $\|\mathbf{m}\|_1$ of $\mathbf{m}$ is $\sum_{i\in[d]}|\mathbf{m}(i)|$, and the {\em $1$-norm} $\|A\|_1$ of an integer matrix $A$ is $\sum_{i,j}|A(i,j)|$.
The $1$-norm $\|\mathbf{u}\|_1$ of $\mathbf{u}\in\mathbb{N}_{\omega}^d$ is defined by $\sum_{i\in[d],\ \mathbf{u}(i)\neq\omega}|\mathbf{u}(i)|$, {\em ignoring the $\omega$ entries}.

\subsection{Non-Elementary Complexity Classes}

Reachability in VASS is not elementary even in fixed dimensions~\cite{czerwinski2021reachability}.
To characterize the problem complexity, one needs complexity classes beyond the elementary classes.
Schmidt introduced an ordinal indexed class of complexity classes $\textbf{F}_3,\textbf{F}_4,\ldots,\textbf{F}_{\omega},\ldots,$ $\textbf{F}_{\omega^2},\ldots,\textbf{F}_{\omega^{\omega}},\ldots$ and showed that many problems arising in theoretical computer science are complete problems in this hierarchy~\cite{schmitz2016complexity}.
In the above sequence $\textbf{F}_{3}=$\textsc{Tower} and $\textbf{F}_{\omega}=\textsc{Ackermann}$.
The class \textsc{Tower} is closed under elementary reductions and \textsc{Ackermann} is closed under primitive recursive reductions. For the purpose of this paper it suffices to say that \textsc{Tower} contains all the problems whose space complexity is bounded by tower functions of the form
\[2^{\iddots ^{2^{n}}\Big\}\,f(n)},\]
where $f(n)$ is an elementary function.

The notation $\texttt{poly}(n)$ will stand for a polynomial bound, and $\texttt{exp}(n)$ an exponential bound.
Most of the time we shall not be explicit about constant factors when making statements about upper bounds.

\subsection{Integer Programming}

We shall need a result in integer linear programming~\cite{pottier1991minimal}. Let $A$ be an $m\,{\times}\,k$ integer matrix and $\mathbf{x}\in\mathbb{V}^k$. The {\em homogeneous equation system} of $A$ is given by the linear equation system $\mathcal{E}$ specified by
\begin{equation}\label{2019-05-15}
A\mathbf{x}=\mathbf{0}.
\end{equation}
A nontrivial solution to~(\ref{2019-05-15}) is some $\mathbf{m}\in\mathbb{N}^k\setminus\{\mathbf{0}\}$ such that $A\mathbf{m}=\mathbf{0}$.
The set of solutions form a monoid $(\mathcal{S},\mathbf{0},+)$.
Since the pointwise ordering $\le$ is a well quasi order on $\mathbb{N}^k$, the set $\mathcal{S}$ must be generated by a finite set of nontrivial minimal solutions.
This finite set is called the {\em Hilbert base} of $\mathcal{E}$, denoted by $\mathcal{H}(\mathcal{E})$.
The following important result is proved by Pottier~\cite{pottier1991minimal}, in which $r$ is the {\em rank} of $A$.
\begin{lemma}[Pottier]\label{pottier-lemma}
$\|\mathbf{m}\|_1 \le (1\,{+}\,k{\cdot}\|A\|_{1})^r$ for every $\mathbf{m}\in\mathcal{H}(\mathcal{E})$.
\end{lemma}

Let $\mathbf{r}\in\mathbb{Z}^k$.
Nonnegative integer solutions to {\em equation system}
\begin{equation}\label{hes}
A\mathbf{x}=\mathbf{r}
\end{equation}
can be derived from the Hilbert base of the homogeneous equation system $A\mathbf{x}-x'\mathbf{r}=\mathbf{0}$.
Let $\mathbb{S}^{=\mathbf{r}}$ be the finite set of the minimal solutions to $A\mathbf{x}-x'\mathbf{r}=\mathbf{0}$ with $x'=1$, and $\mathbb{S}^{=\mathbf{0}}$ be the finite set of the minimal solutions to $A\mathbf{x}-x'\mathbf{r}=\mathbf{0}$ with $x'=0$.
A solution to~(\ref{hes}) is of the form
\[
\mathbf{m} + \sum_{i\in[|\mathbb{S}^{=\mathbf{0}}|]}k_i\mathbf{m}_i,
\]
where $\mathbf{m}\in\mathbb{S}^{=\mathbf{r}}$, $\mathbf{m}_i\in\mathbb{S}^{=\mathbf{0}}$ and $k_i$ is a natural number for each $i\in[|\mathbb{S}^{=\mathbf{0}}|]$.
The following is an immediate consequence of Lemma~\ref{pottier-lemma}.
\begin{corollary}\label{eq-sol}
$\|\mathbf{m}\|_1 \le (1+k{\cdot}\|A\|_{1}+\|\mathbf{r}\|_{1})^{r+1}$ for all $\mathbf{m}\in\mathbb{S}^{=\mathbf{r}}\cup\mathbb{S}^{=\mathbf{0}}$.
\end{corollary}

The size of~(\ref{2019-05-15}) can be defined by $mk\log(\|A\|_{1})$, and the size of~(\ref{hes}) by $mk\log(\|A\|_{1})+\|\mathbf{r}\|_{1}$.
The size of $(1+k{\cdot}\|A\|_{1}+\|\mathbf{r}\|_{1})^{r+1}$ is polynomial.
Thus $|\mathbb{S}^{=\mathbf{r}}|$ and $|\mathbb{S}^{=\mathbf{0}}|$ are bounded by some exponential functions.
In polynomial space a nondeterministic algorithm can guess a solution and check if it is minimal, hence the following.
\begin{corollary}\label{min-solu}
Both $\mathbb{S}^{=\mathbf{r}}$ and $\mathbb{S}^{=\mathbf{0}}$ can be produced in $\texttt{poly}(n)$ space.
\end{corollary}

\subsection{Vector Addition Systems with States}

By a {\em digraph} we mean a finite directed graph in which multi-edges and self loops are admitted.
A {\em $d$-dimensional vector addition system with states}, or {\em $d$-VASS}, is a labeled digraph $G=(Q,T,q_{in},q_{out})$ where $Q$ is the set of vertices and $T$ is the set of edges.
The edges are labeled by elements of $\mathbb{Z}^d$, and the labels are called {\em displacements}.
A {\em state} is identified to a vertex and a {\em transition} is identified to a labeled edge. We write $o,p,q$ for states, $t$ and its decorated versions for edges.
A transition from $p$ to $q$ labeled $\mathbf{t}$ is denoted by $(p,\mathbf{t},q)$ and $p\stackrel{\mathbf{t}}{\longrightarrow}q$.
Two special states are identified, an input state $q_{in}$ and an output state $q_{out}$.
We sometimes abbreviate $(Q,T,q_{in},q_{out})$ to $(Q,T)$ when the input and the output states are not to be referred.

A path $\pi$ from $p_0$ to $q_n$ of $G=(Q,T)$ is a sequence $(p_0,\mathbf{a}_0,q_0),\ldots,(p_n,\mathbf{a}_n,q_n)$ of transitions such that $p_i=q_{i-1}$ for all $i\in [n]$.
The displacement $\Delta(\pi)$ of $\pi$ is $\sum_{i\in[n]_0}\mathbf{a}_i$.
If $p=p_0=q_n$, we call $\pi$ a {\em cycle} of $G$ on $p$.
In the rest of the paper we refer to $G=(Q,T,p_{in},q_{out})$ either as a graph or as a VASS.
Define
\begin{eqnarray}
\|T\| &=& \sum_{p\stackrel{\mathbf{t}}{\longrightarrow}q\in T}\|\mathbf{t}\|_1, \label{2023-07-30-a}\\
\max\|T\| &=& \max\left\{\|\mathbf{t}\|_1\mid p\stackrel{\mathbf{t}}{\longrightarrow}q\in T\right\}. \label{2023-07-30-b}
\end{eqnarray}
The input size $|G|$ of a VASS $G=(Q,T,p,q)$ is defined as follows:
$|G|\defn|Q|+\|T\|$, where $|Q|$ is the size of $Q$.
A Parikh image for $G=(Q,T)$ is a vector in $\mathbb{N}^{T}$. We will write $\phi,\varphi,\psi$ for Parikh images. The displacement $\Delta(\psi)$ is defined by $\sum_{t=(p,\mathbf{t},q)\in T}\psi(t){\cdot}\mathbf{t}$.
For a path $\pi$ in $G$, we define $\wp(\pi)$ as the Parikh image of $\pi$.

Given a space $\mathbb{M}$, a {\em configuration} in $\mathbb{M}$ for the $d$-VASS $G=(Q,T)$ is a pair $(p,\mathbf{m})\in Q\times \mathbb{M}$, often abbreviated to $p(\mathbf{m})$, where $p$ is the state of the configuration and $\mathbf{m}$ is the {\em location} of the configuration.
For $t=(p,\mathbf{a},q)$, we write $p(\mathbf{m})\xrightarrow{t}_{\mathbb{M}}q(\mathbf{n})$
whenever $\mathbf{n}=\mathbf{m}+\mathbf{a}$ and $\mathbf{m},\mathbf{n}\in\mathbb{M}$.
Given a path $\pi=t_1\ldots t_n$ from $p$ to $q$, we say that $\pi$ is a {\em run} in $\mathbb{M}$, written $p(\mathbf{m})\xrightarrow{\pi}_{\mathbb{M}}q(\mathbf{n})$, if there exist configurations $p_1(\mathbf{m}_1),\ldots,p_{n-1}(\mathbf{m}_{n-1})$ such that
$p(\mathbf{m})\xrightarrow{t_1}_{\mathbb{M}}p_1(\mathbf{m}_1)\xrightarrow{t_2}_{\mathbb{M}}\cdots \xrightarrow{t_{n-1}}_{\mathbb{M}}p_{n-1}(\mathbf{m}_{n-1})\xrightarrow{t_n}_{\mathbb{M}}q(\mathbf{n})$.
We write \[p(\mathbf{m})\xrightarrow{G}_{\mathbb{M}}q(\mathbf{n})\]
for the existence of a run $p(\mathbf{m})\xrightarrow{\pi}_{\mathbb{M}}q(\mathbf{n})$ in $G$.
We say that $\pi$ is a {\em walk} if it is a run in $\mathbb{N}^d$.
When talking about walks, we often omit the subscript $\mathbb{N}^d$ and will assume that $\mathbf{m},\mathbf{n}\in\mathbb{N}^d$.
The {\em reachability problem} can be formally stated as follows:
\begin{quotation}
Given a $d$-VASS $G=(Q,T,p,q)$ and two locations $\mathbf{m},\mathbf{n}\in\mathbb{N}^d$,
is $p(\mathbf{m})\xrightarrow{G}q(\mathbf{n})$?
\end{quotation}

For an edge $t\in T$, let $V_G(t)$ be the vector space $V_G(t)\subseteq\mathbb{Q}^d$ spanned by the displacements of the cycles that contain $t$.
Let $V_G$ be the vector space spanned by the displacements of all cycles in $G$.
Let $n_G$ be $dim\,V_G$, which is the dimension of $V_G$.
We say that $G$ is \emph{effectively $n_G$-dimensional}.
\begin{figure}
    \centering
    \includegraphics[scale = 0.9]{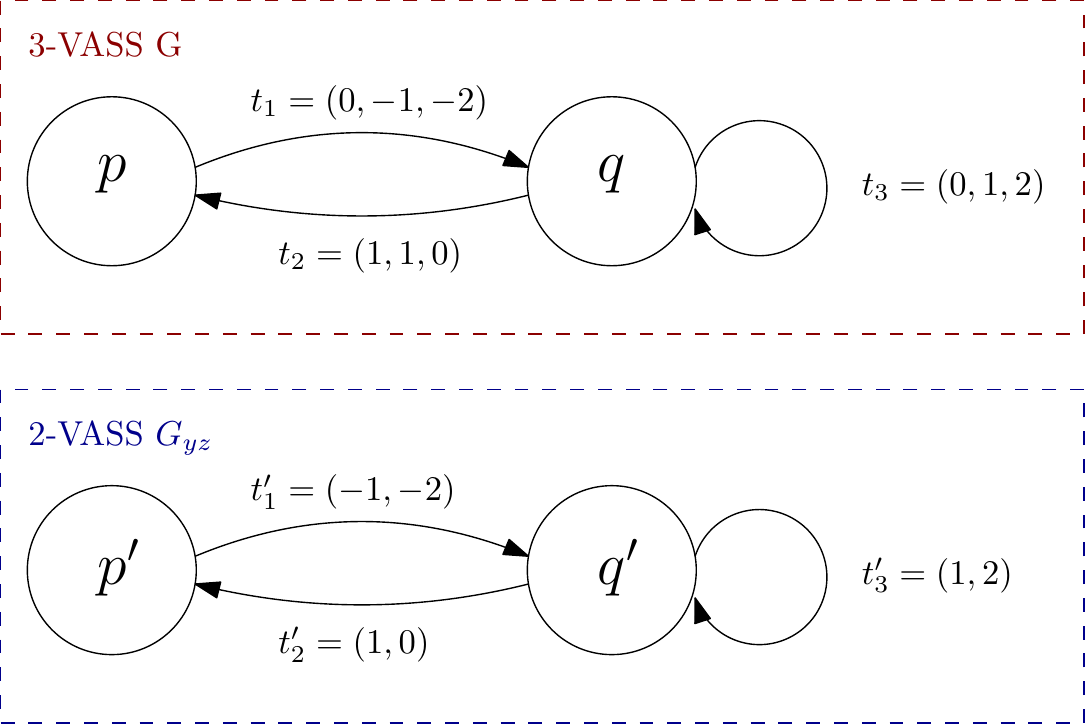}
    \caption{An Example of an effectively $2$-dimensional $3$-VASS}
    \label{fig-eff3vass}
\end{figure}
\begin{example}\label{example:2d3vass}
Consider the $3$-VASS $G$ presented in Figure~\ref{fig-eff3vass}.
It contains two states $p,q$ and three transitions $\mathbf{t}_1,\mathbf{t}_2,\mathbf{t}_3$.
The space $V_G$ is spanned by the displacements of the cycles $(1,0,-2)$, $(0,1,2)$.
So $G$ is an effectively $2$-dimensional $3$-VASS $G$.
It is the hyperplane defined by the equation $2x-2y+z=0$.
\end{example}
In this paper we pay special attention to the effectively $2$-dimensional $3$-VASSes.
In such a VASS all configurations of the form $p(\mathbf{v})$ reachable from $p(\mathbf{u})$ lie in a hyperplane.
This suggests to investigate the relationship between the effectively $2$-dimensional $3$-VASSes and the $2$-VASSes.
We will show in Section~\ref{sec-2023-03-12} that walks in an effectively $2$-dimensional $3$-VASS has the same linear path scheme property as the walks in a $2$-VASS.

We will find it necessary to talk about which part of the $d$-dimensional space a particular vector $\mathbf{m}$ is in.
Let $\#_1,\#_2,\ldots,\#_d\in\{\ge,\le\}$.
The {\em zone} $Z_{(\#_1,\#_2,\ldots,\#_d)}$ is defined the equivalence:
$\mathbf{m}\in Z_{(\#_1,\#_2,\ldots,\#_d)}$ if and only if $\mathbf{m}(1)\#_1 0$, $\mathbf{m}(2)\#_2 0$, \ldots, and $\mathbf{m}(d)\#_d0$.
For instance the quadrant $\mathbb{N}\times(-\mathbb{N})$ is denoted by $Z_{(\ge,\le)}$, and the octant $\mathbb{N}^3$ is denoted as $Z_{(\geq,\geq,\geq)}$.
For a $d$-dimensional vector $\mathbf{o}$, write $Z_{\mathbf{o}}$ for the zone in which $\mathbf{o}$ lies.

\section{Effectively $2$-Dimensional $3$-VASS}
\label{sectionlinearpathscheme}

It has been known for some time that walks in $2$-VASS can be described by linear path schemes~\cite{blondin2015reachability}.
We will point out in this section a non-surprising generalization of this result, that is all walks in an effectively $2$-dimensional $3$-VASS are describable by linear path schemes.
We review in Section~\ref{2023-03-11} the characterization of walks in terms of linear path schemes due to Blondin {\em et~al}~\cite{blondin2015reachability}, and carry out the generalization in Section~\ref{sec-2023-03-12}.
We fix a VASS $G=(Q,T,p,q)$ throughout the section.

\subsection{Linear Path Scheme}\label{2023-03-11}

Suppose $G=(Q,T)$ is a $d$-dimensional VASS.
A linear path scheme (from $p$ to $q$) is a regular expression of the form $\rho=\alpha_0\beta_1^*\alpha_1\cdots\beta_n^*\alpha_n$, where $\alpha_1,\ldots,\alpha_n$ are paths in $G$ and $\beta_1,\ldots,\beta_n$ are cycles in $G$.
The length $|\rho|$ of $\rho$ is defined by $|\alpha_0\beta_1\alpha_1\cdots\beta_n\alpha_n|$.
For two locations $\mathbf{m},\mathbf{n}$, we write $p(\mathbf{m})\xrightarrow{\rho}q(\mathbf
n)$ if there exist $e_1,\ldots e_n\in\mathbb{N}$ such that $p(\mathbf{m})\xrightarrow{\alpha_0\beta_1^{e_1}\ldots\beta_n^{e_n}\alpha_n}q(\mathbf{n})$ is a {\em walk}.
Intuitively $\rho$ defines a walk pattern dominated by the circular walks $\beta_1,\ldots,\beta_n$.
For every $c\in \mathbb{N}$, let $\mathscr{L}_{c}(G)$ be the set of the linear path schemes whose length is bounded by $(|Q|+\|T\|)^c$.
We are particularly interested in the linear path schemes in which the displacements of all cycles are in $Z_{(\#_1,\#_2,\ldots,\#_d)}$.
\begin{definition}
Suppose $G=(Q,T)$ is a VASS.
A linear path scheme $\rho=\alpha_0\beta_1^*\alpha_1\cdots\beta_n^*\alpha_n$ is {\em zigzag free in $Z_{(\#_1,\#_2,\ldots,\#_d)}$} if $\Delta(\beta_1),\ldots,\Delta(\beta_n)\in Z_{(\#_1,\#_2,\ldots,\#_d)}$.
\end{definition}
A linear path scheme $\rho$ is zigzag free if it is zigzag free in some $Z_{(\#_1,\#_2,\ldots,\#_d)}$.
To shed more light on the zigzag free property, define the following in terms of the constants in~(\ref{2023-07-30-a}) and~(\ref{2023-07-30-b}).
\begin{eqnarray}
\mathcal{D} &\stackrel{\rm def}{=}& \max\|T\|{\cdot}(|Q|+\|T\|)^c, \label{D-2023-07-27-a} \\
\mathbb{D} &\stackrel{\rm def}{=}& [2\mathcal{D},\infty). \label{D-2023-07-27-b}
\end{eqnarray}
The number $\mathcal{D}$ depends only on $G$.
If a path starts from a configuration $p(\mathbf{m})$ in  $\mathbb{D}^3$ and is of length bounded by $(|Q|+\|T\|)^c$, then it must be completely in the first octant.
Now suppose $\mathbf{m},\mathbf{n}\in\mathbb{D}^d$, and that $\alpha_0\beta_1^*\alpha_1\cdots\beta_n^*\alpha_n\in\mathscr{L}_{c}(G)$ is a linear path scheme that is zigzag free in $Z_{\mathbf{n}-\mathbf{m}}$.
One has the property stated in the next lemma.
\begin{lemma}\label{LZY2023-07-28}
If $\pi=\alpha_0\beta_1^{e_1}\alpha_1\cdots\beta_n^{e_n}\alpha_n$ is a path from $p(\mathbf{m})$ to $q(\mathbf{n})$, then $p(\mathbf{m})\xrightarrow{\pi}q(\mathbf{n})$.
\end{lemma}
\begin{proof}
Recall that $\mathbf{1}$ is the vector whose entries are all $1$.
It follows from the inequality $\mathbf{m}\ge2\mathcal{D}\mathbf{1}$ and the definition of $\mathcal{D}$ that $\mathcal{D}\mathbf{1}\le\mathbf{m}+\Delta(\alpha_1)+\ldots+\Delta(\alpha_i)$ for every $i\in[n]$.
By the zigzag free property of $\beta_1,\ldots,\beta_n$, one derives that for every $i\in[n]$, the inequality $\mathbf{m}+\Delta(\alpha_0\beta_1^{e_1}\alpha_1\cdots \beta_i^{e_i})\ge\mathcal{D}\mathbf{1}$ holds.
By the length bound on $\beta_1,\ldots,\beta_n$, the path $\alpha_0\beta_1^{e_1}\alpha_1\cdots \beta_i^{e_i}$, starting from $\mathbf{m}$, lies completely in $Z_{(\ge,\ge,\ldots,\ge)}$.
\end{proof}

The following proposition and its corollary, due to
Blondin {\em et~al}~\cite{blondin2015reachability}, reveal the importance of the zigzag free linear path schemes to the $2$-VASS reachability problem.
\begin{proposition}\label{theorem2linearpathscheme}
Suppose $G=(Q,T)$ is a $2$-dimensional VASS.
There exists a constant $c$ depending on $G$ such that, for all $\mathbf{m},\mathbf{n}\in\mathbb{D}^2$, the following statements are valid.
 \begin{enumerate}
 \item If $q(\mathbf{m})\rightarrow_{\mathbb{N}^2}q(\mathbf{n})$, then there exists a linear path scheme $\rho\in\mathscr{L}_c(G)$ zigzag free in $Z_{\mathbf{n}-\mathbf{m}}$ with at most two cycles such that $q(\mathbf{m})\xrightarrow{\rho}q(\mathbf{n})$.
 \item If $p(\mathbf{m})\rightarrow_{\mathbb{D}^2}q(\mathbf{n})$, then there exists a linear path scheme $\rho\in\mathscr{L}_c(G)$ with at most $2|Q|$ cycles such that $p(\mathbf{m})\xrightarrow{\rho}q(\mathbf{n})$.
     \end{enumerate}
\end{proposition}

The interested reader is advised to consult~\cite{blondin2015reachability} for the long proof of Proposition~\ref{theorem2linearpathscheme}.
It is worth emphasizing that in the first statement of Proposition~\ref{theorem2linearpathscheme} the vector $\mathbf{n}-\mathbf{m}$ can be in any one of the four quadrants.
The locations $\mathbf{m},\mathbf{n}$ can be far away.
If $p(\mathbf{m})\xrightarrow{G}q(\mathbf{n})$ is captured by $p(\mathbf{m})\xrightarrow{\alpha_0\beta_1^*\alpha_1\beta_1^*\alpha_2}q(\mathbf{n})$, the role of the transitions $\alpha_0,\alpha_1,\alpha_2$ is to adjust positions, it takes a linear combination of the circular walks $\beta_1,\beta_2$ to take us from $p(\mathbf{m})$ to $q(\mathbf{n})$.

With the help of this proposition a close relationship between the walks in $2$-dimensional VASSes and the zigzag free linear path schemes can be established as in the following corollary.
\begin{corollary}\label{linearpathscheme-2023-07-18}
Suppose $G=(Q,T)$ is a $2$-dimensional VASS.
There exists a constant $c$ depending on $G$ such that, for all $\mathbf{m},\mathbf{n}\in\mathbb{N}^2$, there is a walk $p(\mathbf{m})\xrightarrow{G}q(\mathbf{n})$ if and only if $p(\mathbf{m})\xrightarrow{\rho}q(\mathbf{n})$ for some $\rho\in\mathscr{L}_c(G)$.
\end{corollary}
In Section~\ref{sec-2023-03-12} we will prove a generalization of Corollary~\ref{linearpathscheme-2023-07-18}.
We will say that the walk $p(\mathbf{m})\xrightarrow{G}q(\mathbf{n})$ is {\em captured} by a linear path scheme $\rho\in\mathscr{L}_c(G)$ if $p(\mathbf{m})\xrightarrow{\rho}q(\mathbf{n})$.
Notice that there are at most $|T|^{(|Q|+\|T\|)^c}$ linear path schemes in $\mathscr{L}_c(G)$, and that there are infinitely many pairs of configuration $p(\mathbf{m}),q(\mathbf{n})$ such that $p(\mathbf{m})\xrightarrow{G}q(\mathbf{n})$.
No matter how apart $\mathbf{m},\mathbf{n}$ are, a walk $p(\mathbf{m})\xrightarrow{G}q(\mathbf{n})$ is captured by some $\rho$ whose length is independent of $\mathbf{m},\mathbf{n}$.

A distinguished feature of the linear path schemes is that they do not contain any nested cycles.
In literature this is called the flatness property.
The flatness allows one to enforce the nonnegativity condition on the linear path schemes by linear (in)equations.
In the following definition $\mathbf{x},\phi,\mathbf{y},\left\{\mathbf{z}^1_{k,l}\right\}_{l\in[|\alpha_k|]},\left\{\mathbf{z}^2_{k,l}\right\}_{l\in[|\beta_k|]},\left\{\mathbf{z}^3_{k,l}\right\}_{l\in[|\beta_k|]}$ are vectors of variables.
\begin{definition}\label{LPSS}
    Let $\rho=\alpha_0\beta_1^*\alpha_1\beta_2^*\cdots \beta_n^*\alpha_n\in\mathscr{L}_c(G)$ be in a $d$-dimensional VASS $G$.
    The {\em linear path scheme system} (LPS system) $\mathscr{E}_{\rho}$ for
    $\mathbf{x},\mathbf{y}\in\mathbb{N}^d$ consists of the following equations.
\begin{eqnarray}
    \label{lps-displacement}
    \mathbf{x}+\sum_{i=0}^n\Delta(\alpha_i)+
    \sum_{i=1}^n\phi(\beta_i){\cdot}\Delta(\beta_i) &=& \mathbf{y}, \\
    \label{lps-alpha}
    \mathbf{x}+\sum_{i=1}^k(\Delta(\alpha_{i-1})+\phi(\beta_i)\Delta(\beta_i))+
    \Delta(\alpha_{k}[1,\ldots,l]) &=& \mathbf{z}^1_{k,l}, \\
    \label{lps-beta1}
    \mathbf{x}+\sum_{i=1}^k(\Delta(\alpha_{i-1})+\phi(\beta_i)\Delta(\beta_i))+
    \Delta(\alpha_{k})+\Delta(\beta_{k+1}[1,\ldots,l]) &=& \mathbf{z}^2_{k+1,l}, \\
    \label{lps-betan}
    \mathbf{x}+\sum_{i=1}^{k+1}(\Delta(\alpha_{i-1})+
    \phi(\beta_i)\Delta(\beta_i))-     \Delta(\beta_{k+1}[l,\ldots,|\beta_{k+1}|]) &=& \mathbf{z}^3_{k+1,l}.
\end{eqnarray}
\end{definition}
In the above definition $0\le k\le n$ and $0\le l\le|\alpha_k|$ in~(\ref{lps-alpha}), $0\le k\le n-1$ and $1\le l\le|\beta_{k+1}|$ in~(\ref{lps-beta1}) and~(\ref{lps-betan}).
To check if a path of the form $\alpha_0\beta_1^{e^1}\alpha_1\beta_2^{e_2}\ldots\beta_{n}^{e_n}\alpha_n$ is a walk from $p(\mathbf{m})$ to $q(\mathbf{n})$, one only has to verify the following three conditions:
\begin{itemize}
    \item The displacement of the path is $\mathbf{n}-\mathbf{m}$. This is equation~(\ref{lps-displacement}).
    \item For every prefix $\alpha_k'$ of $\alpha_k$, $\alpha_0\beta_1^{e^1}\alpha_1\beta_2^{e_2}\ldots\beta_{k-1}^{e_{k-1}}\alpha_k'$ is a walk.
        This is equation~(\ref{lps-alpha}).
    \item For every cycle $\beta_k$, both the paths throughout the first lap and the last lap are in the first octant.
        This is the equations~(\ref{lps-beta1}) and~(\ref{lps-betan}).
\end{itemize}
We will write more briefly $\mathbf{y}=\mathscr{E}_{\rho}(\mathbf{x})$ for $\mathscr{E}_{\rho}$.
Suppose $\mathbf{f}$ is a solution to $\mathbf{y}=\mathscr{E}_{\rho}(\mathbf{x})$.
We write for example $\mathbf{f}(\mathbf{x})$ for the vector assigned to $\mathbf{x}$ by the solution $\mathbf{f}$.

The homogeneous LPS system $\mathscr{E}^0_{\rho}$ is defined by the following equations.
\begin{eqnarray}
        \mathbf{x}^0+\sum_{i=1}^n\phi^0(\beta_i)\Delta(\beta_i)&=&\mathbf{y}^0, \\
        \mathbf{x}^0+\sum_{i=1}^k\phi^0(\beta_i)\Delta(\beta_i)&=&\mathbf{z}^{0,1}_{k,l}, \\
        \mathbf{x}^0+\sum_{i=1}^k\phi^0(\beta_i)\Delta(\beta_i)&=&\mathbf{z}^{0,2}_{k+1,l}, \\
        \mathbf{x}^0+\sum_{i=1}^{k+1}\phi^0(\beta_i)\Delta(\beta_i)&=&\mathbf{z}^{0,3}_{k+1,l}.
\end{eqnarray}
Again we write more briefly $\mathbf{y}^0=\mathscr{E}^0_{\rho}(\mathbf{x}^0)$ for $\mathscr{E}^0_{\rho}$.
Let $\mathbf{f}^0$ be a solution to $\mathbf{y}^0=\mathscr{E}^0_{\rho}(\mathbf{x}^0)$.
The summation $\mathbf{f}+h\mathbf{f}^0$ is a solution to $\mathscr{E}_{\rho}$ for all $h\in\mathbb{N}$.

\begin{proposition}\label{LPSCh}
Every linear path scheme $\rho$ is characterized by a system $\mathscr{E}_{\rho}$ of linear Diophantine equations in the sense that $\mathscr{E}_{\rho}$ has a solution if and only if $p(\mathbf{\mathbf{m}})\xrightarrow{\rho}q(\mathbf{n})$.
\end{proposition}
\begin{proof}
Suppose $\rho=\alpha_0\beta_1^*\alpha_1\cdots\beta_n^*\alpha_n$ is a linear path scheme of $G=(Q,T,p,q)$ that starts with $p$ and ends with $q$.
Suppose $\mathbf{f}$ is a solution to $\mathbf{y}=\mathscr{E}_{\rho}(\mathbf{x})$ and $\mathbf{m}=\mathbf{f}(\mathbf{x})$ and $\mathbf{n}=\mathbf{f}(\mathbf{y})$.
We claim that the path defined by a solution $\mathbf{f}$ is a walk from $p(\mathbf{m})$ to $q(\mathbf{n})$.
Let $\pi=\alpha_0\beta_1^{\mathbf{f}(\phi(\beta_1))}\alpha_1\cdots\beta_n^{\mathbf{f}(\phi(\beta_n))}\alpha_n$.
We only need to prove by induction that no configurations in the path fall outside $\mathbb{N}^d$.
This is done by induction.
Suppose $p(\mathbf{m})\xrightarrow{\alpha_0\beta_1^{\mathbf{f}(\phi(\beta_1))}\alpha_1\cdots\beta_k^{\mathbf{f}(\phi(\beta_k))}}p_k(\mathbf{m}_k)$ holds for $k$.
There are two cases to consider.
  \begin{itemize}
    \item Since the equation (\ref{lps-alpha}) holds, one has, for all $l\in[|\alpha_{k}|]$, the following
    \[\mathbf{m}_k+\Delta(\alpha_{k}[1,\ldots,l])=\mathbf{m}+\sum_{i=1}^{k+1}(\Delta(\alpha_{i-1})+\mathbf{f}(\phi(\beta_i))\Delta(\beta_i))+
    \Delta(\alpha_{k}[1,\ldots,l])=\mathbf{f}(\mathbf{z}^1_{k,l})\ge\mathbf{0}.\]
So the run of $\alpha_{k}$ that starts with $p_k(\mathbf{m}_k)$ never drops below $0$.
    \item Let the configuration after $\alpha_{k}$ be $p'_k(\mathbf{m}'_k)$, the equations~(\ref{lps-beta1}) and~(\ref{lps-betan}) show that, for all $l\in[|\beta_{k+1}|]$, one has the following inequations.
    \[
        \mathbf{m}'_{k}+\Delta(\beta_{k+1}[1,\ldots,l])=\mathbf{f}(\mathbf{z}^2_{k+1,l})\ge\mathbf{0},
    \]
    \[
       \mathbf{m}'_k+(\mathbf{f}(\phi(\beta_{k+1}))-1)\Delta(\beta_{k+1})+\Delta(\beta_{k+1}[1,\ldots,l])=\mathbf{f}(\mathbf{z}^3_{k+1,l})\ge\mathbf{0}.
    \]
    The first shows that $p'_k(\mathbf{m}'_k)$ can perform the first cycle $\beta_{k+1}$ and the second guarantees that $p'_k(\mathbf{m}'_k+(\mathbf{f}(\phi(\beta_{k+1}))-1)\Delta(\beta_{k+1}))$ can do the last cycle $\beta_{k+1}$.
    By monotonicity $p_k(\mathbf{m}'_k)$ can perform the cycle $\beta_{k+1}$ for $\mathbf{f}(\phi(\beta_{k+1}))$ times.
    \end{itemize}
The implication in the other direction is clear.
\end{proof}

\subsection{Walks in Effectively $2$-Dimensional $3$-VASSes}\label{sec-2023-03-12}

Let $G$ be an effectively $2$-dimensional $3$-VASS.
We will show that walks between two configurations in $G$ can be captured by linear path schemes.
The proof follows that of Corollary~\ref{linearpathscheme-2023-07-18} of Blondin~\cite{blondin2015reachability}.
To start with we prove that regional walks between two configurations can be captured by linear path schemes.
We will establish two facts.
\begin{enumerate}
    \item For large polynomial $C(n)$ every walk with both the start location and the final location in $[C(|G|),+\infty)^3$ can be captured by a linear path scheme.
    This is Proposition~\ref{Theroembig3vass}.
    \item For every polynomial $C(n)>0$ every walk with both the start location and the final location bounded by $C(|G|)$ in at least one dimension can be captured by a linear path scheme.
        This is Lemma~\ref{theorembounded3vass} and lemma~\ref{theoremcombbounded3vass}.
\end{enumerate}
Now fix some large polynomial $C(n)$.
A walk in the first octant generally passes through both regions.
We will show that the walk can be decomposed into a polynomial number of local walks, each of them falls in either $[C(|G|),+\infty)^3$ or $[C(|G|)]_0{\times}\mathbb{N}^2\cup \mathbb{N}{\times}[C(|G|)]_0{\times}\mathbb{N} \cup \mathbb{N}^2{\times}[C(|G|)]_0$.
So the whole walk can be captured by a linear path scheme with a polynomial number of cycles.
We will actually let the two regions to overlap so that the start locations and the final locations of the local walks fall in the overlapping space.

The following is a generation of Proposition~\ref{theorem2linearpathscheme}.
\begin{proposition}\label{Theroembig3vass}
Suppose $G=(Q,T)$ is an effectively $2$-dimensional $3$-VASS.
There exists a constant $c$ depending on $G$ such that, for all $\mathbf{m},\mathbf{n}\in\mathbb{D}^3$, the following statements are valid.
 \begin{enumerate}
 \item If $q(\mathbf{m})\rightarrow_{\mathbb{N}^3}q(\mathbf{n})$, then there exists a linear path scheme $\rho\in\mathscr{L}_c(G)$ zigzag free in $Z_{\mathbf{n}-\mathbf{m}}$ with at most two cycles such that $q(\mathbf{m})\xrightarrow{\rho}q(\mathbf{n})$.
 \item If $p(\mathbf{m})\rightarrow_{\mathbb{D}^3}q(\mathbf{n})$, then there exists a linear path scheme $\rho\in\mathscr{L}_c(G)$ with at most $2|Q|$ cycles such that $p(\mathbf{m})\xrightarrow{\rho}q(\mathbf{n})$.
     \end{enumerate}
\end{proposition}

Proposition~\ref{Theroembig3vass} is proved by projecting an effectively $2$-dimensional $3$-VASS onto an axis plane.
By Proposition~\ref{theorem2linearpathscheme} and Corollary~\ref{linearpathscheme-2023-07-18} a walk in a $2$-VASS can be captured by a linear path scheme.
We hope to be able to restore the $3$-dimensional linear path scheme from such a $2$-dimensional linear path scheme in the projected $2$-VASS.
The question is onto which axis plane should the hyperplane be projected.
To shed more light on the question, let's take a look at an example.
Consider a $3$-VASS with single state $p$ and two transitions $\mathbf{t}_1=(1,0,-1)$ and $\mathbf{t}_2=(0,1,1)$.
We ask if $p(150,200,100)$ is reachable from $p(50,50,50)$.
If we project the $3$-VASS onto the first two dimensions, we can simply get a linear path scheme $(t_1')^{100}(t_2')^{150}$, where $t_1'=(1,0)$ and $t_2'=(0,1)$, that captures a walk from $p(50,50)$ to $p(150,200)$.
However $(t_1)^{100}(t_2)^{150}$ is not a walk in the $3$-VASS because the values in the third dimension goes into the negative somewhere in the path.
Instead of projecting onto the first and the second entries, let's project a $3$-tuple onto its first and third entries.
It is not difficult to see that all walks from $p(50,50)$ to $p(150,100)$ are captured by linear path schemes.
For example $(\mathbf{t}_1''\mathbf{t}_2'')^{100}\left(\mathbf{t}_2''\right)^{50}$ is one such linear path scheme, where $t_1''=(1,-1)$ and $t_2''=(0,1)$.
In the latter case the sign of the values in the first and the third dimensions determine the sign of the second entry.
As long as the first and the third entries are non-negative, so is the second entry.
The example should help understand the next lemma and its proof.
\begin{lemma}\label{lemma2dimvectorspace}
    Let $V\subseteq \mathbb{Q}^3$ be a $2$-dimensional subspace.
    If the vector space spanned by $V\cap Z_{(\#_1,\#_2,\#_3)}$ is $V$, then there exist distinct $i,j\in[3]$ such that  $\mathbf{m}\in Z_{(\#_1,\#_2,\#_3)}$ whenever $\mathbf{m}\in V$, $\mathbf{m}(i)\#_i 0$ and $\mathbf{m}(j)\#_j 0$.
\end{lemma}
\begin{proof}
Suppose $V\cap Z_{(\#_1,\#_2,\#_3)}$ is proper in the sense that $V\cap Z_{(\#_1,\#_2,\#_3)}$ spans $V$.
Suppose there are $a\ne0$, $b\ne0$ and $c\ne0$ such that $V$ is the hyperplane defined by the equation
    \begin{equation}\label{2023-07-18}
        ax+by+cz = 0.
    \end{equation}
The intersection of this hyperplane with the hyperplane $x=0$, respectively the hyperplane $y=0$, and the hyperplane $z=0$ are respectively the following lines:
\begin{eqnarray}
by+cz = 0 & & \text{on the YOZ axis plane}, \label{2023-08-08-a} \\
ax+cz = 0 & & \text{on the XOZ axis plane}, \label{2023-08-08-b} \\
ax+by = 0 & & \text{on the XOY axis plane} \nonumber.
\end{eqnarray}
Since the intersection $V\cap Z_{(\#_1,\#_2,\#_3)}$ is proper, the hyperplane~(\ref{2023-07-18}) must intersect with two axis planes in the zone $Z_{(\#_1,\#_2,\#_3)}$.
Without loss of generality assume that $V$ intersects with~(\ref{2023-08-08-a}) and~(\ref{2023-08-08-b}) in the zone $Z_{(\#_1,\#_2,\#_3)}$.
Let
\begin{equation}
I_{yz} = \left(0,1_{\#_2},-1_{\#_2}{\cdot}\frac{b}{c}\right),\ I_{xz}=\left(1_{\#_1},0,-1_{\#_1}{\cdot}\frac{a}{c}\right)
\end{equation}
be the direction vectors with one unit length entry, where $1_{\#_1},1_{\#_2}$ are defined as follows:
\[
1_{\#_1} = \left\{\begin{array}{ll}
1, & \text{if }\#_1\text{ is }\ge, \\
-1, & \text{if }\#_1\text{ is }\le,
\end{array}\right.
\text{ respectively, }
1_{\#_2} = \left\{\begin{array}{ll}
1, & \text{if }\#_2\text{ is }\ge, \\
-1, & \text{if }\#_2\text{ is }\le,
\end{array}\right.
\]
Notice that by definition $\left(-1_{\#_2}{\cdot}\frac{c}{b}\right)\#_3 0$ and $\left(-1_{\#_1}{\cdot}\frac{c}{a}\right)\#_3 0$.
Then for all $\mathbf{m}\in V$, the inequalities $\mathbf{m}(1)\#_1 0$ and $\mathbf{m}(3)\#_3 0$ imply $\mathbf{m}(2)\#_2 0$.
This is because $\mathbf{m}$ can be represented as $dI_{xy}+d'I_{yz}$ for some nonnegative $d,d'\in\mathbb{Q}$.
But then
    \begin{equation}\label{yd2023-07-28}
    \mathbf{m}=dI_{xy}+d'I_{yz}=\left(1_{\#_1}{\cdot}d',1_{\#_2}{\cdot}d,-d{\cdot}1_{\#_2}{\cdot}\frac{c}{b}-d'{\cdot}1_{\#_1}{\cdot}\frac{c}{a}\right)\in Z_{(\#_1,\#_2,\#_3)}.
    \end{equation}
If in~(\ref{2023-07-18}) the conditions $a\ne0$, $b\ne0$ and $c\ne0$ are not satisfied, the proof is easier.
\end{proof}

Lemma~\ref{lemma2dimvectorspace} is really about the promotion of the zigzag free property.
Suppose $V$ is a $2$-dimensional subspace, and $i=1,j=2$.
Suppose $\beta$ is a $3$-dimensional cycle such that $\Delta(\beta)\in V$, $\Delta(\beta)(1)\#_1 0$ and $\Delta(\beta)(2)\#_2 0$, meaning that the projection of the cycle $\beta$ onto the $XOY$ plane is zigzag free in the quadrant $(\#_1,\#_2)$.
Then $\beta(3)\#_3 0$ can be interpreted as saying that $\beta$ is zigzag free in the octant $Z_{(\#_1,\#_2,\#_3)}$.
With this observation we are able to prove Proposition~\ref{Theroembig3vass}.

\begin{proof}[Proof of Proposition~\ref{Theroembig3vass}]
Suppose the dimension of the space spanned by $V_G\cap Z_{\mathbf{n}-\mathbf{m}}$ is $2$.
The situation is trivial if the dimension is strictly less than $2$.
To prove the first proposition, let $Z_{\mathbf{n}-\mathbf{m}}=(\#_1,\#_2,\#_3)$.
By Lemma~\ref{lemma2dimvectorspace} we may assume without loss of generality that for all $\mathbf{w}\in V_G$, $\mathbf{w}(3)\#_3 0$ whenever $\mathbf{w}(1)\#_1 0$ and $\mathbf{w}(2)\#_2 0$.
Let $V_G$ be spanned by $(1,0,a), (0,1,b)$ for some $a,b\in\mathbb{Q}$.
Let $f:\mathbb{Z}^3\rightarrow \mathbb{Z}^2$ be the projection function defined by $f((x_1,x_2,x_3))= (x_1,x_2)$.
Using the linear equality that defines the subspace $V_G$, it is easy to define the inverse function $f^{-1}$.
The $2$-VASS $G_{xy}=(Q_{xy},T_{xy})$ is defined by
\begin{itemize}
\item $Q_{xy}=Q$, and
\item $T_{xy}=\{(p,f(\mathbf{a}),q)\mid(p,\mathbf{a},q)\in T\}$.
\end{itemize}
Let $\mathcal{D}_{xy}$ and $\mathbb{D}_{xy}$ be defined like in~(\ref{D-2023-07-27-a}) and~(\ref{D-2023-07-27-b}) respectively for the $2$-VASS $G_{xy}$.
Because $\mathcal{D}_{xy}\le\mathcal{D}$, we may use $\mathcal{D}$ and $\mathbb{D}$ defined for $G$ instead of $\mathcal{D}_{xy}$ and $\mathbb{D}_{xy}$ in the following argument.
Since $q(\mathbf{m})\xrightarrow{\pi}_{\mathbb{N}^3}q(\mathbf{n})$ in $G$ for some path $\pi$, we have $q(f(\mathbf{\mathbf{m}}))\xrightarrow{\pi'}_{\mathbb{N}^2}q(f(\mathbf{n}))$ in $G_{xy}$ obtained by projection.
We may write $\pi'= f(\pi)$ by extending $f$ to the projection function that maps $T^*$ onto $T_{xy}^*$ in the obvious manner.
So we may write $q(f(\mathbf{\mathbf{m}}))\xrightarrow{f(\pi)}_{\mathbb{N}^2}q(f(\mathbf{n}))$.
Since $f(\mathbf{m}),f(\mathbf{n})\in \mathbb{D}^2$, according to Proposition~\ref{theorem2linearpathscheme}, there exists a zigzag free linear path scheme $\rho_{xy}=\alpha_0\beta_1^*\alpha_1\beta_2^*\alpha_2\in\mathscr{L}_{c}(G)$ such that $q(f(\mathbf{m}))\xrightarrow{\rho_{xy}}q(f(\mathbf{n}))$.
The zigzag-free property ensures that both $\Delta(\beta_1)$ and $\Delta(\beta_2)$ belong to the quadrant $(\#_1,\#_2)$.
By Lemma~\ref{lemma2dimvectorspace} the displacements of the cycles $f^{-1}(\beta_1)$, $f^{-1}(\beta_2)$ are zigzag free in the octant $Z_{\mathbf{n}-\mathbf{m}}$.

Let $\pi_{xy}=\alpha_0\beta_1^{e_1}\alpha_1\beta_2^{e_2}\alpha_2$ be such that $q(f(\mathbf{m}))\xrightarrow{\pi_{xy}}q(f(\mathbf{n}))$.
Evidently $\Delta(\pi_{xy})=\Delta(\pi')=\Delta(f(\pi))=f(\Delta(\pi))$.
    We also need to prove $\Delta(f^{-1}(\pi_{xy}))=\Delta(\pi)$, which boils down to checking the equality in the third dimension.
    This is essentially due to the fact that the values in the third entry are linear combinations of the values in the first and the second entries.
    If $\mathbf{m}(3)=f^{-1}((\mathbf{m}(1),\mathbf{m}(2)))$, then
    \[
    \Delta(f^{-1}(\pi_{xy}))(3) =a\Delta(\pi_{xy})(1) + b\Delta(\pi_{xy})(2) =a\Delta(\pi)(1) + b\Delta(\pi)(2) =\Delta(\pi)(3).
    \]
    If $\mathbf{m}(3)\not=f^{-1}((\mathbf{m}(1),\mathbf{m}(2)))$, the shift $\mathbf{m}(3)-f^{-1}((\mathbf{m}(1),\mathbf{m}(2)))$ is maintained throughout $f^{-1}(\pi_{xy})$.
    Because $\mathbf{m}\in\mathbb{D}$, the shift does not cause the path to go out of the first octant.
    We still need to prove that $q(\mathbf{m})\xrightarrow{f^{-1}(\pi_{xy})}q(\mathbf{n})$, that is the path is in the first octant.
    But this is just Lemma~\ref{LZY2023-07-28}.
The proof of the first proposition is completed.

The second proposition is proved by making good use of the first proposition.
Suppose $p(\mathbf{m})\rightarrow_{\mathbb{D}^3}q(\mathbf{n})$.
There exists a walk
    \[
        p(\mathbf{m})\xrightarrow{\alpha_0}q_1(\mathbf{m}_1)\xrightarrow{\beta_1}q_1(\mathbf{m}'_1)\xrightarrow{\alpha_1}\cdots q_k(\mathbf{m}_k)\xrightarrow{\beta_k}q_k(\mathbf{m}'_k)\xrightarrow{\alpha_{k+1}}q(\mathbf{n})
        \]
    such that
    \begin{itemize}
        \item $\mathbf{m}_i,\mathbf{m}'_i\in\mathbb{D}^3$, and for all $i\in[k]$,
        \item $|\alpha_i|\leq |Q|$, $q_i(\mathbf{m}_i)$ is the first configuration in which $q_i$ appears and $q_i(\mathbf{m}'_i)$ is the last configuration in which $q_i$ appears.
    \end{itemize}
    By construction $k\leq |Q|$.
    Applying the first proposition, each $\beta_i$ can be captured by a linear path scheme $\alpha_{i,0}\beta^*_{i,1}\alpha_{i,1}\beta^*_{i,2}\alpha_{i,2}$. Thus the walk can be captured by a linear path scheme $\rho$ where $|\rho|\leq (|Q|+\|T\|)^{O(1)}$, and $\rho$ has no more than $2|Q|$ cycles.
\end{proof}

Let's demonstrate the combined power of Lemma~\ref{lemma2dimvectorspace} and Proposition~\ref{Theroembig3vass} by an example.
Consider the $3$-VASS defined in Figure~\ref{fig-eff3vass}.
We look for a linear path scheme that captures a walk from $p(\mathbf{u})$ to $p(\mathbf{v})$, where $\mathbf{u}=(22,22,22)$ and $\mathbf{v}=(42,42,22)$.
The intersection $V_G\cap Z_{\mathbf{v}-\mathbf{u}}$ is above the $YOZ$ plane.
Therefore $\mathbf{x}(2)>0$ and $\mathbf{x}(3)<0$ imply $\mathbf{x}(1)>0$ for every $\mathbf{x}\in V_G$.
We obtain the $2$-VASS $G_{yz}$ by projecting $G$ onto the $YOZ$ plane.
The projection function is $f(\mathbf{x})=(\mathbf{x}(2),\mathbf{x}(3))$.
Let $\mathbf{t}_1',\mathbf{t}_2',\mathbf{t}_3'$ be the projections of $\mathbf{t}_1,\mathbf{t}_2,\mathbf{t}_3$ respectively.
The following path $\pi$ is from $p(f(\mathbf{u}))$ to $q(f(\mathbf{v}))$:
\[
p(22,22)\xrightarrow{\mathbf{t}'_1}q(21,20)\xrightarrow{(\mathbf{t}'_3)^{20}}q(41,60)\xrightarrow{\mathbf{t}'_2}p(42,60)\xrightarrow{(\mathbf{t}'_1\mathbf{t}'_2)^{19}}p(42,22).
\]
The displacement of $\mathbf{t}'_3$ and $\mathbf{t}'_1\mathbf{t}'_2$ are zigzag-free.
Consequently the displacement of $\mathbf{t}_3$ and $\mathbf{t}_1\mathbf{t}_2$ are zigzag-free.
It is clear that $p(\mathbf{m})$ can engage in the walk $\mathbf{t}_1(\mathbf{t}_3)\mathbf{t}_2(\mathbf{t}_1\mathbf{t_2})$. Therefore the following is a walk.
\[p(22,22,22)\xrightarrow{\mathbf{t}_1}q(22,21,20)\xrightarrow{(\mathbf{t}_3)^{20}}q(22,41,60)\xrightarrow{\mathbf{t}_2}p(23,42,60)\xrightarrow{(\mathbf{t}_1\mathbf{t}_2)^{19}}p(42,42,22).\]

Next we consider the walks that are close to an axis plane.
We shall repeat an argument of~\cite{blondin2015reachability} to the effectively $2$-dimensional $3$-VASS.
Set $\mathbb{I}_1=[2\mathscr{D}]_0\times\mathbb{N}^2$, $\mathbb{I}_2=\mathbb{N}\times[2\mathscr{D}]_0\times\mathbb{N}$, and $\mathbb{I}_3=\mathbb{N}^2\times[2\mathscr{D}]_0$.
Let $\,\mathbb{I}_{\mathscr{D}}$ stand for one of $\mathbb{I}_1$, $\mathbb{I}_2$, $\mathbb{I}_3$.
Consider a walk in $\,\mathbb{I}_{\mathscr{D}}$.
Values in the bounded dimension are encoded into the states as it were.
A state becomes a pair $(q,t)$ with $q\in Q$ and $t\in[2\mathscr{D}]_0$.
This operation transforms the effectively $2$-dimensional $3$-VASS to a $2$-VASS, the reachability of the latter can be characterized by linear path schemes according to Proposition~\ref{theorem2linearpathscheme}, hence the following lemma.
\begin{lemma}\label{theorembounded3vass}
    Suppose $\mathbf{m},\mathbf{n}\in \mathbb{I}_{\mathscr{D}}$.
    Then
    $p(\mathbf{m})\rightarrow_{\mathbb{I}_{\mathscr{D}}}q(\mathbf{n})$ if and only if there exists a linear path scheme $\rho$ with size $|\rho|\leq (|Q|+\|T\|)^{O(1)}$ such that $p(\mathbf{m})\xrightarrow{\rho}q(\mathbf{n})$.
\end{lemma}
\begin{proof}
Consider the case that $\mathbb{I}_{\mathscr{D}}=[2\mathscr{D}]_0\times\mathbb{N}^2$. Define a $2$-VASS $\widehat{G}=(\widehat{Q},\widehat{T})$ as follows:
    \begin{eqnarray*}
        \widehat{Q} &\defn& \left\{q_i\mid q\in Q,\ i\in [2\mathscr{D}]_0\right\}, \\
        \widehat{T} &\defn& \left\{\left(p_i,(\mathbf{a}(2),\mathbf{a}(3)),q_{i+\mathbf{a}(1)}\right)\mid (p,\mathbf{a},q)\in T\text{ and } i,i+\mathbf{a}(1)\in[2\mathscr{D}]_0\right\}.
    \end{eqnarray*}
    A simple induction shows that for every pair $\mathbf{m},\mathbf{n}\in \mathbb{I}_{\mathscr{D}}$, the following equivalence is valid.
    \begin{eqnarray}
        \text{``}p(\mathbf{m})\rightarrow_{\mathbb{I}_{\mathscr{D}}}q(\mathbf{n})\ in \ G\text{''}\text{ iff } \text{``}p_{\mathbf{m}(1)}((\mathbf{m}(2),\mathbf{m}(3)))\rightarrow_{\mathbb{N}^2} q_{\mathbf{n}(1)}((\mathbf{n}(2),\mathbf{n}(3)))\ in\ \widehat{G}\text{''}.\label{eqn3vassreduction}
    \end{eqnarray}
    By Proposition~\ref{theorem2linearpathscheme} there exists a linear path scheme $\widehat{\rho}$ with size $|\widehat{\rho}|\leq (|\widehat{Q}|+\|\widehat{T}\|)^{O(1)}$ such that $p_{\mathbf{m}(1)}((\mathbf{m}(2),\mathbf{m}(3)))\xrightarrow{\widehat{\rho}}q_{\mathbf{n}(1)}((\mathbf{n}(2),\mathbf{n}(3)))$ in $\widehat{G}$.
    Evidently we can transfer $\widehat{\rho}$ back to a linear path scheme in $G$ whose length is bounded by $(|Q|+\|T\|)^{O(1)}$.
\end{proof}

Having proved Lemma~\ref{theorembounded3vass}, we consider next the situation where a walk may go from one of the regions $\mathbb{I}_1,\mathbb{I}_2,\mathbb{I}_3$ to another.
Let $\mathbb{L}_{\mathscr{D}}$ be the union $\mathbb{I}_1\cup\mathbb{I}_2\cup\mathbb{I}_3$.
The following lemma shows that the walks in $\mathbb{L}_{\mathscr{D}}$ can also be captured by linear path schemes.
\begin{lemma}\label{theoremcombbounded3vass}
Suppose $G$ is an effectively $2$-dimensional $3$-VASS and $\mathbf{m},\mathbf{n}\in \mathbb{L}_{\mathscr{D}}$.
Then
    $p(\mathbf{m})\rightarrow_{\mathbb{L}_{\mathscr{D}}}q(\mathbf{n})$ if and only if there exists a linear path scheme $\rho$ with size $|\rho|\leq (|Q|+\|T\|)^{O(1)}$ such that $p(\mathbf{m})\xrightarrow{\rho}q(\mathbf{n})$.
\end{lemma}
\begin{proof}
If the vector space $V_G$ is parallel to one of the coordinate planes, say the $XOY$-plane, then the effectively $2$-dimensional $3$-VASS is essentially a $2$-VASS.
We are done by applying Proposition~\ref{Theroembig3vass}.
  In the rest of the proof we assume that $V_G$ is not parallel to any coordinate plane.
  Let
  \[
  \mathbb{J}_{ij}\stackrel{\rm def}{=}\mathbb{I}_i\cap\mathbb{I}_j,\text{ where } i<j.
  \]
These are the regions of crossing points.
  A walk $\pi$ in $\mathbb{L}_{\mathscr{D}}$ from $p(\mathbf{m})$ to $q(\mathbf{n})$ may pass many crossing points.
  It is generally of the following form
      \begin{equation}\label{2023-06-08}
          p(\mathbf{m})\xrightarrow{\pi_1}_{\mathbb{B}_1}p_1(\mathbf{m}_1)\xrightarrow{\pi_2}_{\mathbb{B}_2}p_2(\mathbf{m}_2)\xrightarrow{\pi_3}_{\mathbb{B}_3}\cdots \xrightarrow{\pi_k}_{\mathbb{B}_{k}}p_k(\mathbf{m}_k)\xrightarrow{\pi_{k+1}}_{\mathbb{B}_{k+1}}q(\mathbf{n}),
      \end{equation}
  where
      \begin{itemize}
          \item $\mathbb{B}_i\in\{\mathbb{I}_1,\mathbb{I}_2,\mathbb{I}_3\}$ all $i\in[k+1]$,
          \item $\mathbb{B}_{i+1}\neq\mathbb{B}_i$ for all $i\in[k]$, and
          \item $\mathbf{m}_i\in \{\mathbb{J}_{12},\mathbb{J}_{13},\mathbb{J}_{23}\}$ for all $i\in[k]$.
      \end{itemize}
  If $k\leq 3{\cdot}|Q|{\cdot} (2\mathscr{D}\,{+}\,1)^2$, then since by Lemma~\ref{theorembounded3vass} each $\pi_j$ can be captured by a linear path scheme of length bounded by $(|Q|+\|T\|)^{O(1)}$, the path $\pi$ itself is then captured by a linear path scheme of length bounded by $(|Q|+\|T\|)^{O(1)}$.
  Next suppose $k>3{\cdot}|Q|{\cdot} (2\mathscr{D}\,{+}\,1)^2$.
  By the pigeon hole principle there must exist $i<j$ such that $p_i=p_j$ and at least two of the equalities $\mathbf{m}_i(1)=\mathbf{m}_j(1)$, $\mathbf{m}_i(2)=\mathbf{m}_j(2)$, $\mathbf{m}_i(3)=\mathbf{m}_j(3)$ are valid.
  Without loss of generality assume that $\mathbf{m}_i(1)=\mathbf{m}_j(1)$ and $\mathbf{m}_i(2)=\mathbf{m}_j(2)$.
  The existence of the cycle implies that for some $a$,
  \begin{equation}\label{2023-07-01}
  \mathbf{a}\stackrel{\rm def}{=}\mathbf{m}_j-\mathbf{m}_i=(0,0,a)\in V_G.
  \end{equation}
  If also $|\{i\mid \mathbf{m}_i\in\mathbb{J}_{13}\}|> |Q|{\cdot} (2\mathscr{D}\,{+}\,1)^2$, then by the same argument there would be some circular walk whose displacement is of the form $(0,b,0)\in V_G$ with $b\ne0$.
  This is a contradiction because the existence of $\mathbf{a}$ and $(0,b,0)$ would imply that $V_G$ is parallel to the coordinate plane $Y{O}Z$.
  The same line of reasoning applies to the set $\{i\mid \mathbf{m}_i\in\mathbb{J}_{23}\}$ as well.
  Therefore
  \begin{eqnarray}
  |\{i\mid \mathbf{m}_i\in\mathbb{J}_{13}\}| &\le& |Q|{\cdot} (2\mathscr{D}\,{+}\,1)^2, \label{20230-6-09-a}\\
  |\{i\mid \mathbf{m}_i\in\mathbb{J}_{23}\}| &\le& |Q|{\cdot} (2\mathscr{D}\,{+}\,1)^2. \label{20230-6-09-b}
  \end{eqnarray}
 It follows from the inequalities~(\ref{20230-6-09-a}) and~(\ref{20230-6-09-b}) that $|\{i\mid \mathbb{B}_i=\mathbb{I}_3\}|\le|Q|{\cdot}(2\mathscr{D}\,{+}\,1)^2$. Let the set $\{i\mid \mathbb{B}_i=\mathbb{I}_3\}$ be $\{i_1,\ldots,i_l\}$ such that $i_1<\ldots<i_l$.
 Assume that $\pi$ does not contain any sub-walk from a configuration to itself.
By setting $p_1(\mathbf{o}_1)=p(\mathbf{m})$ and $p_{i_l+1}'(\mathbf{o}_{i_l+1}')=q(\mathbf{n})$, the walk $\pi$ can be rearranged into the following form:
  \begin{equation}\label{rearrangedwalk}
  p_1(\mathbf{o}_1)\xrightarrow{\pi'_1}_{\mathbb{I}_1\cup\mathbb{I}_2}p_1'(\mathbf{o}_1')\xrightarrow{\pi_{i_1}}_{\mathbb{I}_3}p_2(\mathbf{o}_2)\xrightarrow{\pi'_2}_{\mathbb{I}_1\cup\mathbb{I}_2}\cdots \xrightarrow{\pi_{i_{l}}}_{\mathbb{I}_3}p_{i_{l}}(\mathbf{o}_{i_{l}})
  \xrightarrow{\pi_{l+1}'}_{\mathbb{I}_1\cup\mathbb{I}_2}p'_{i_{l+1}}(\mathbf{o}_{i_{l+1}}').
  \end{equation}
  We need to bound $|\pi_1'|,\ldots,|\pi_{l+1}'|$.
  Suppose $s\in[l+1]$ and that $\pi'_s$ lies completely in $\mathbb{I}_1$.
  Let \[\mathscr{D}' \stackrel{\rm def}{=} |Q|{\cdot} (2\mathscr{D}\,{+}\,1){\cdot}\|T\|.\]
  We prove that $\pi'_s$ stays completely in the region
  \begin{eqnarray*}
  \mathbb{K} &\stackrel{\rm def}{=}& [2\mathscr{D}]_0\times\left[\max\{0,\mathbf{o}_s(2)-\mathscr{D}'\},\mathbf{o}_s(2)+\mathscr{D}'\right]\times\mathbb{N}.
  \end{eqnarray*}
  Suppose
  \begin{equation}\label{2023-07-20}
  p_{s}(\mathbf{o}_{s})\xrightarrow{\varpi}_{\mathbb{I}_{1}}p'(\mathbf{o}') \xrightarrow{\varpi'}_{\mathbb{I}_{1}}p''(\mathbf{o}'')
  \end{equation}
  is an initial sequence of $\pi'_s$ such that $p_{s}(\mathbf{o}_{s})\xrightarrow{\varpi}_{\mathbb{I}_{1}}p'(\mathbf{o}')$ only passes the locations in $\mathbb{K}$ and $\mathbf{o}''(2)>\mathbf{o}_s(2)+\mathscr{D}'$.
Remove from~(\ref{2023-07-20}) sub-paths of the form $\overline{p}(\overline{\mathbf{o}})\xrightarrow{\overline{\pi}}_{\mathbb{I}_{1}}\overline{p}(\overline{\mathbf{o}}')$
      such that $\overline{\mathbf{o}}(2)=\overline{\mathbf{o}}'(2)$.
In the path so obtained there must be a sub-path $\underline{p}(\underline{\mathbf{o}})\xrightarrow{\underline{\pi}}_{\mathbb{I}_{1}}\underline{p}(\underline{\mathbf{o}}')$
      such that $\underline{\mathbf{o}}(1)=\underline{\mathbf{o}}'(1)$.
Then $\underline{\mathbf{o}}-\underline{\mathbf{o}}'$ is some $(0,b,c)$ such that $b\ne0$.
Together with the vector $\mathbf{a}$ in~(\ref{2023-07-01}) it implies that $V_G$ is parallel to the coordinate plane $Y{O}Z$, contradicting to the assumption.
For the same reason $\mathbf{o}''(2)<\mathbf{o}_s(2)-\mathscr{D}'$ is impossible.
We conclude that $\pi'_{s}$ never goes outside $\mathbb{K}$.
By the same argument one shows that if $\pi'_s$ lies completely in $\mathbb{I}_2$, it must stay within the region $\left[\max\{0,\mathbf{o}_s(1)-\mathscr{D}'\},\mathbf{o}_s(1)+\mathscr{D}'\right]\times[2\mathscr{D}]_0\times\mathbb{N}$.
It is now easy to see that if $\pi_s'$ goes from $\mathbb{I}_1$ to $\mathbb{I}_2$, or vice versa, then it must lie in the region
\begin{eqnarray*}
\mathbb{L}_{12} &\stackrel{\rm def}{=}& \left([2\mathscr{D}]_0\times\left[\mathscr{D}'\right]_0\times\mathbb{N}\right) \cup \left([\mathscr{D}']_0\times[2\mathscr{D}]_0\times\mathbb{N}\right).
\end{eqnarray*}

  Based on the above discussion, the walk $\pi$ can be segmented in the following form:
  \begin{equation}\label{rearranged_walk2}
      p(\mathbf{m})\xrightarrow{\pi''_1}_{\mathbb{C}_1}p_1(\mathbf{m}'_1)\xrightarrow{\pi''_2}_{\mathbb{C}_2}p_2(\mathbf{m}'_2)\xrightarrow{\pi_3''}_{\mathbb{C}_3}\cdots \xrightarrow{\pi''_{k'-1}}_{\mathbb{C}_{k'-1}}p_{k'-1}(\mathbf{m}_{k'-1})\xrightarrow{\pi''_{k'}}_{\mathbb{C}_{k'}}q(\mathbf{n}),
  \end{equation}
  where
  \begin{itemize}
      \item $\mathbb{C}_i\in\{\mathbb{I}_1,\mathbb{I}_2,\mathbb{I}_3,\mathbb{L}_{12}\}$ all $i\in[k'+1]$,
      \item $\mathbb{C}_{i+1}\neq\mathbb{C}_i$ for all $i\in[k']$, and
      \item $\mathbf{m}_i\in \{\mathbb{J}_{13},\mathbb{J}_{23},\mathbb{J}_{12}\}$ for all $i\in[k']$.
  \end{itemize}
  Notice that by~(\ref{20230-6-09-a}) and~(\ref{20230-6-09-b}), the number of the occurrences of $\mathbb{I}_3$ is at most $|Q|{\cdot} (2\mathscr{D}\,{+}\,1)^2$.
 Consequently $k'\leq 2{\cdot}|Q|{\cdot}(2\mathscr{D}\,{+}\,1)^2$.
  By Lemma~\ref{theorembounded3vass}, each sub-walk $\pi_j$ in $\mathbb{I}_1,\mathbb{I}_2,\mathbb{I}_3$ can be captured by a linear path scheme of length bounded by $(|Q|+\|T\|)^{O(1)}$, noticing that $\mathscr{D}'$ is also $(|Q|+\|T\|)^{O(1)}$.
  The region $\mathbb{L}_{12}$ is essentially $1$-dimensional, so it can be captured by a linear path scheme of length bounded by $(|Q|+\|T\|)^{O(1)}$.
  Consequently the whole walk $\pi$ can be captured by a concatenation of at most $2{\cdot}|Q|{\cdot}(2\mathscr{D}\,{+}\,1)^2$ linear path schemes, each bounded by $(|Q|+\|T\|)^{O(1)}$ in length.
  The proof is complete.
\end{proof}

We have proved that a walk in an effectively $2$-dimensional $3$-VASS can be converted to a linear path scheme if either (i) the first location and the last location are high up in the first octant or (ii) one dimension is restricted.
For a general walk we divide the first octant into two regions:
\begin{itemize}
\item $\mathbb{D}$, and
\item $\mathbb{L}\defn[2\mathscr{D}+2\mathcal{T}]_0\,{\times}\, \mathbb{N}^2\cup\mathbb{N}\,{\times}\,[2\mathscr{D}+2\mathcal{T}]_0\,{\times}\,\mathbb{N}\cup \mathbb{N}^2\,{\times}\,[2\mathscr{D}+2\mathcal{T}]_0$, where $\mathcal{T}=\max\|T\|{\cdot}|Q|$.
\end{itemize}
Suppose $p(\mathbf{m})\rightarrow q(\mathbf{n})$ and consider all the configurations in the walk.
Let $q(\mathbf{m}_{q,in})$ and $q(\mathbf{m}_{q,out})$ in the walk be the first and the last configurations whose states are the same $q$ and whose locations are both in $\mathbb{D}\cap\mathbb{L}$.
Using this strategy, the walk can be segmented to consecutive sub-paths of two categories:
\begin{enumerate}
    \item circular walks from $q(\mathbf{m}_{q,in})$ to $q(\mathbf{m}_{q,out})$ for $q\in Q$ and $\mathbf{m}_{q,in},\mathbf{m}_{q,out}\in\mathbb{D}\cap\mathbb{L}$, and
    \item walks between the cycles.
\end{enumerate}
The walks in the first category can be converted into linear path schemes by Proposition~\ref{Theroembig3vass}.
The walks in the second category are bounded by $|Q|$ in length.
They are either completely in $\mathbb{D}$ or completely in $\mathbb{L}$, and can be transformed into linear path schemes by Proposition~\ref{Theroembig3vass} and respectively by a variant of Lemma~\ref{theoremcombbounded3vass} with $[2\mathscr{D}]_0$ replaced by $[2\mathscr{D}+2\mathcal{T}]_0$.
All the linear path schemes are bounded in size by $(|Q|+\|T\|)^{O(1)}$.
Hence the main result of the section.
\begin{theorem}
    \label{theorem3vasslps}
    Suppose $G$ is an effectively $2$-dimensional $3$-VASS $G$.
     If $p(\mathbf{m})\xrightarrow{G}q(\mathbf{n})$, then there exists a linear path scheme $\rho$ of size $(|Q|+\|T\|)^{O(1)}$ such that $p(\mathbf{m})\xrightarrow{\rho}q(\mathbf{n})$.
\end{theorem}

To check if there is a walk from $p(\mathbf{m})$ to $q(\mathbf{n})$, an algorithm only has to guess a linear path scheme of size $(|Q|+\|T\|)^{O(1)}$ and check if the associated LPS system has any solution.

We are in a position to prove Theorem~\ref{cor:reachabilityforeffectively2dim3vass}.
\begin{proof}[Proof of Theorem~\ref{cor:reachabilityforeffectively2dim3vass}]
Let $G=(Q,T,p,q)$ be an effectively $2$-dimensional $3$-VASS and $p(\mathbf{m}),q(\mathbf{n})$ be two configurations.
By Theorem~\ref{theorem3vasslps}, $p(\mathbf{m})\xrightarrow{G}q(\mathbf{n})$ if and only if $p(\mathbf{m})\xrightarrow{\rho}q(\mathbf{n})$ for some linear path scheme $\rho$ with $|\rho|\leq (|Q|+\|T\|)^{O(1)}$.
By Proposition~\ref{LPSCh} the LPS system $\mathbf{n}=\mathscr{E}_{\rho}(\mathbf{m})$ has a solution $\mathbf{f}$ and by Corollary~\ref{min-solu} the size of $\mathbf{f}$ is $(|Q|+\|T\|)^{O(|Q|+\|T\|)}$.
Thus a walk $p(\mathbf{m})\xrightarrow{G}q(\mathbf{n})$ can be guessed in exponential space.
\end{proof}

\section{KLMST Constructions}\label{sectionklmstalgo}

We recall in this section the main constructions in the KLMST method.
Our account follows~\cite{leroux2019reachability} closely.
All the results stated in this section are from~\cite{leroux2019reachability}.
We omit the proofs.
Special attention is paid to the treatment to the effectively $2$-dimensional VASSes since that is where our algorithm improves upon the KLMST algorithm in the 3-dimensional case.
Throughout the section we fix a VASS $G=(Q,T,q_{in},q_{out})$.

\subsection{KLM sequence}\label{sec-KLM-sequence-4-8}

A $d$-dimensional {\em KLM sequence} is a sequence of connected VASSes.
It is of the form
\begin{equation}\label{2023-04-17}
        \xi=(\mathbf{u}_0G_0\mathbf{v}_0)\mathbf{a}_1(\mathbf{u}_1G_1
        \mathbf{v}_1)\mathbf{a}_2\cdots \mathbf{a}_n(\mathbf{u}_nG_n
        \mathbf{v}_n),
\end{equation}
where
\begin{itemize}
    \item $\mathbf{u}_i,\mathbf{v}_i\in \mathbb{N}_{\omega}^{d}$ for all $i\in[n]_0$,
    \item $G_i=(Q_i,T_i,p_i,q_i)$ is a $d$-VASS for all $i\in[n]_0$, and
    \item $\mathbf{a}_i\in \mathbb{N}^d$ is the displacement of a transition from $q_{i-1}$ to $p_{i}$ for all $i\in[n]$.
        We shall call $\mathbf{a}_1,\ldots,\mathbf{a}_n$ the {\em connecting edges}.
\end{itemize}
The component $\mathbf{u}_iG_i\mathbf{v}_i$ is denoted by $\xi_i$.
If $G_i$ is the trivial VASS that contains one state and no transition, then $(\mathbf{u}_i,G_i,\mathbf{v}_i)$ is vacuous.
For convenience the transition whose displacement is $\mathbf{a}_i$ is often referred to by $\mathbf{a}_i$.
The size of $\xi$ is $2(d+1)^{d+1}(n+\sum_{i=0}^n(\|\mathbf{x}_i\|_1+|G_i|+\|\mathbf{y}_i\|_1)+\sum_{i=1}^n\|\mathbf{a}_i\|_1$), denoted by $|\xi|$.
A run $\rho$ of $\xi$ is of the form $\pi_0\mathbf{a}_1\pi_1\cdots \mathbf{a}_n\pi_n$
where $\pi_i$ is a path from $p_i$ to $q_i$ inside $G_i$.
The run $\rho$ is a {\em witness} to $\xi$ if there exists a sequence of locations $\mathbf{m}_0,\mathbf{n}_0,\ldots,\mathbf{m}_n,\mathbf{n}_n\in\mathbb{N}^{d}$ such that
\begin{itemize}
    \item $\mathbf{m}_i\sqsubseteq \mathbf{u}_i,\mathbf{n}_i\sqsubseteq\mathbf{v}_i$ for all $i\in [n]_0$, and
    \item $p_0(\mathbf{m}_0)\xrightarrow{\pi_0}q_0(\mathbf{n}_0)\xrightarrow{\mathbf{a}_1}p_1(\mathbf{m}_1)\xrightarrow{\pi_1}\cdots \xrightarrow{\mathbf{a}_n}p_n(\mathbf{m}_n)\xrightarrow{\mathbf{\pi_n}}q_n(\mathbf{n}_n)$ is a walk.
\end{itemize}
Let $\textsf{W}_{\xi}$ be the set of the witnesses to $\xi$.
If $\mathbf{u}_0,\mathbf{v}_n\in\mathbb{N}^{d}$, a witness to $\xi$ is also called a walk from $\mathbf{u}_0$ to $\mathbf{v}_n$.

We see the triple $\mathbf{u}G\mathbf{v}$ as an instance of the {\em generalized reachability problem} $\mathbb{VASS}^d$.
\begin{definition}
A triple $\mathbf{u}G\mathbf{v}$ is in $\mathbb{VASS}^d$ if $\textsf{W}_{\xi}\ne\emptyset$.
\end{definition}

The characteristic system $\mathcal{E}_{\xi}$ for $\xi$ as given in~(\ref{2023-04-17}) is a linear Diophantine system that provides an algebraic characterization of $\textsf{W}_{\xi}$.
Formally $\mathcal{E}_{\xi}$ consists of the following equations:
\begin{eqnarray}
    \mathbf{y}_i &=& \mathbf{x}_i+\sum_{t=(p,\mathbf{a},q)\in T_i}\phi_i(t){\cdot} \mathbf{a}, \label{characteristicequation}\\
    \mathbf{x}_{j+1} &=& \mathbf{y}_j+\mathbf{a}_{j+1}, \label{characteristicequation2} \\
    \mathbf{1}_{q_i}-\mathbf{1}_{p_i} &=& \sum_{t=(p,\mathbf{a},q)\in T_i}\phi_i(t){\cdot}(\mathbf{1}_q-\mathbf{1}_p), \label{characteristickirhoffsystem}\\
    \mathbf{x}_i &\sqsubseteq& \mathbf{u}_i, \label{characteristicmodel-entry} \\
    \mathbf{y}_i &\sqsubseteq& \mathbf{v}_i, \label{characteristicmodel}
\end{eqnarray}
where $i\in[n]_0$ and $j\in[n-1]_0$.
The variables $\mathbf{x}_i$ respectively the variables $\mathbf{y}_i$ are introduced for the input locations respectively the output locations of $\xi_i$.
The variables $\phi_i$ are for the number of edges in $G_i$.
In~(\ref{characteristickirhoffsystem}) the notation $\mathbf{1}_{q}$ for example is an indicator vector whose $q$-th entry is $1$ and whose other entries are $0$.
The equality(\ref{characteristickirhoffsystem}) is called \emph{Euler Condition}, which guarantees the existence of a path in $G_i$ from $p_i$ to $q_i$ whose Parikh image is $\phi_i$.
In~(\ref{characteristicmodel-entry}) and~(\ref{characteristicmodel}) the order relation imposes constraints on the input location and the output location.
If say $\mathbf{u}_i(j)=9$, there is an equation $\mathbf{x}_i(j)=9$.
If $\mathbf{u}_i(j)=\omega$, then the constraint is vacuous for $\mathbf{x}_i(j)$.
We shall only be interested in {\em non-negative} integer solutions to $\mathcal{E}_{\xi}$.
Whenever we say ``solution'' we mean ``non-negative integer solution''.
It is clear that if $\mathcal{E}_{\xi}$ is not \emph{satisfiable}, meaning that $\mathcal{E}_{\xi}$ has no solution, then $\xi$ has no witness.
In the rest of the paper we only consider KLM sequences whose characteristic systems are satisfiable.
Our nondeterministic algorithm terminates whenever it comes across an unsatisfiable KLM sequence.
Satisfiability will not be mentioned most of the time.

Let $\mathbf{f}$ be a solution to $\mathcal{E}_{\xi}$.
We denote by $\mathbf{f}(\mathbf{x}_i)$ the vector assigned to $\mathbf{x}_i$ by the solution $\mathbf{f}$.
The notations $\mathbf{f}(\mathbf{y}_i)$ and $\mathbf{f}(\phi_i)$ are interpreted in the same fashion.
The size of $\mathbf{f}$ is defined by $\sum_{i=0}^n(\|\mathbf{f}(\mathbf{x}_i)\|_1+\|\mathbf{f}(\mathbf{y}_i)\|_1+\|\mathbf{f}(\phi_i)\|_1)$.
Now let $V_{\xi}=\bigcup_{i\in[n]_0}(\{\mathbf{x}_i(k)\ |\ k\in[d]\}\cup\{\mathbf{y}_i(k)\ |\ k\in[d]\}\cup \{\phi_i(t)\ |\ t\in T_i\})$, which is the set of the variables of the equation system $\mathcal{E}_{\xi}$.
For each $z\in V_{\xi}$ let $S_{z}$ be the set $\{\mathbf{f}(z)\mid \mathbf{f}\text{ is a solution to }\mathcal{E}_{\xi}\}$.
We say that $z$ is \emph{unbounded} if $S_{z}$ is infinite.

Suppose $\mathbf{f},\mathbf{f}'$ are solutions to $\mathcal{E}_{\xi}$.
We say that $\mathbf{f}'$ is {\em over} $\mathbf{f}$ if $\mathbf{f}\le\mathbf{f}'$ by the point-wise order.
If $\mathbf{f}'$ is over $\mathbf{f}$, then $\mathbf{f}'-\mathbf{f}$ is a solution to the homogeneous characteristic system $\mathcal{E}^0_{\xi}$ defined by the following equations:
\begin{eqnarray}
    \label{homocharacteristicequation}
    \mathbf{y}^0_i&=&\mathbf{x}^0_i+\sum_{t=(p,\mathbf{a},q)\in T_i}\phi^0_i(t){\cdot}\mathbf{a}, \\
    \label{homocharacteristicequation2}
    \mathbf{x}^0_{j+1}&=&\mathbf{y}^0_j, \\
    \label{homocharacteristickirhoffsystem}
    \mathbf{0}&=&\sum_{t=(p,\mathbf{a},q)\in T_i}\phi^0_i(t){\cdot}(\mathbf{1}_q-\mathbf{1}_p), \\
    \label{homocharacteristicmodel}
    \mathbf{x}^0_i(k)&=&0, \text{ whenever } \mathbf{m}_i[k]\neq \omega, \\
    \label{homocharacteristicmodel2}
    \mathbf{y}^0_i(k)&=&0, \text{ whenever } \mathbf{n}_i[k]\neq \omega,
\end{eqnarray}
where $i\in[n]_0$ and $j\in[n-1]_0$ and $k\in[d]$.
Let $V_{\xi}^0=\bigcup_{i\in[n]_0}(\{\mathbf{x}^0_i(k)\ |\ k\in[d]\}\cup\{\mathbf{y}^0_i(k)\ |\ k\in[d]\}\cup \{\phi^0_i(t)\ |\ t\in T_i\})$.
The size of a homogeneous solution is defined similarly.
The variables in $V_{\xi},V_{\xi}^0$ are in 1-1 correspondence.
Set
\begin{equation}\label{2023-07-21}
\mathbf{h}^0=\sum\mathcal{H}(\mathcal{E}^0_{\xi}).
\end{equation}
Clearly $\mathbf{h}^0$ is a solution to $\mathcal{E}^0_{\xi}$.
The following lemma~\cite{leroux2019reachability} is immediate from definition.
\begin{lemma}\label{2023-04-15}
$\mathbf{h}^0(z^0)>0 $ if and only if $z\in V_{\xi}$ is unbounded.
Moreover the sum of the values of the bounded variables is less than $|\xi|^{|\xi|-1}$.
\end{lemma}

A remarkable measure on the KLM sequences was proposed by Leroux and Schmitz~\cite{leroux2019reachability}.
The {\em ranking function} $r$ maps a VASS $G$ onto a $(d{+}1)$-dimensional vector $(r_d,r_{d-1},\ldots,r_0)$, where $r_k$ denotes the number of edges $t$ in $G$ satisfying $dim\ V_{G}(t)=k$.
The rank $r(\xi)$ of the KLM sequence~(\ref{2023-04-17}) is defined by the summation $r(\xi)=\sum_{i\in[n]}r(G_i)$.
It has been shown that the constructions of the KLMST algorithm strictly decrease the rank of any KLM sequence~\cite{leroux2019reachability}, hence the termination of the KLMST algorithm.
In the rest of the section we are going to explain three constructions relevant to our algorithm for the $3$-VASS reachability problem.

\subsection{Standardization}\label{secdecomp}

The KLMST algorithm builds on the fact that a KLM sequence can be converted to a `good' one so that further treatment to the KLM sequence is smooth.
A KLM sequence $\xi$ as in~(\ref{2023-04-17}) is \emph{strongly connected} if for every $i\in[n]_0$ the graph $G_i$ in $\xi_i$ is strongly connected;
it is \emph{saturated} if for every $i\in[n]_0$ and every $j\in[d]$, the equality $\mathbf{u}_i(j)=\omega$, respectively the equality $\mathbf{v}_i(j)=\omega$, is valid if and only if $\mathbf{x}_i(j)$, respectively $\mathbf{y}_i(j)$, is unbounded.
\begin{definition}
A KLM sequence $\xi$ is \emph{standard} if it is \emph{strongly connected} and \emph{saturated}.
\end{definition}
If $\xi_i=\mathbf{u}_iG_i\mathbf{v}_i$ is not saturated, one can update $\mathbf{u}_i,\mathbf{v}_i$ using the solution $\mathbf{h}^0$ of Lemma~\ref{2023-04-15}.
If $G_i$ is not strongly connected, one can find out the strongly connected components of $G_i$  and linearize as it were the strongly connected components of $G_i$.
Notice that there may well be several ways to linearize the strongly connected components.
A nondeterministic algorithm has to guess such a linearization.
By the Euler condition a linearization of $\xi_i=\mathbf{u}_iG_i\mathbf{v}_i$ must be of the form
\begin{equation}\label{2023-07-24}
(\mathbf{u}_{i,1}G_{i,1}\mathbf{v}_{i,1}) \mathbf{a}_{i,2} (\mathbf{u}_{i,2}G_{i,2}\mathbf{v}_{i,2})\mathbf{a}_{i,3}\ldots \mathbf{a}_{i,i_n}(\mathbf{u}_{i,i_n}G_{i,i_n}\mathbf{v}_{i,i_n}),
\end{equation}
where $\mathbf{u}_{i,1}=\mathbf{u}_{i}$ and $\mathbf{u}_{i,i_n}=\mathbf{v}_{i}$.

Using these simple manipulations one can prove the following lemma~\cite{leroux2019reachability}.
\begin{lemma}\label{theroemstandard}
A set $\Xi$ of standard KLM sequences of smaller rank can be computed from a nonstandard $\xi$ in $\mathsf{exp}(|\xi|^{|\xi|})$ time such that $\mathsf{W}_{\xi}=\bigcup_{\xi'\in\Xi}\mathsf{W}_{\xi'}$ and $|\xi'|\le|\xi|^{|\xi|}$ for all $\xi'\in \Xi$.
\end{lemma}
In \cite{leroux2019reachability} the standard KLM sequences are called \emph{clean} KLM sequences, and Lemma~\ref{theroemstandard} is called Cleaning Lemma.

\subsection{Decomposition}\label{sectiondimenreduct}

A KLM sequence $\xi$ as in~(\ref{2023-04-17}) is said to be \emph{unbounded} if for every $i\in[n]_0$ and every edge $t$ of $G_i=(Q_i,T_i)$, the set $S_{\phi_i(t)}$ is unbounded; it is \emph{bounded} otherwise.
According to Lemma~\ref{2023-04-15} the unboundedness is completely determined by the solution $\mathbf{h}^0$ to the homogeneous characteristic system.
Even if $G_i$ is strongly connected, there may be some $t\in T_i$ such that $S_{\phi_i(t)}$ is finite.
One can construct the strongly connected components of the {\em unbounded} edges and linearize these components.
In this way the edge $t$ occurs as a connecting edge $\phi_i(t)$ times.
If for example $G_i$ is bounded, the component $\xi_i=\mathbf{u}_iG_i\mathbf{v}_i$ can be {\em decomposed} to a KLM sequence of the form~(\ref{2023-07-24}).
There can be many decompositions of $\xi_i$.
A nondeterministic algorithm has to guess one of them.
The next lemma is also from~\cite{leroux2019reachability}.
\begin{lemma}\label{theorembounded}
    A set $\Xi$ of  KLM sequences of smaller rank can be computed from a bounded standard $\xi$ in $\mathsf{exp}(|\xi|^{|\xi|})$ time such that $\mathsf{W}_{\xi}=\bigcup_{\xi'\in\Xi}\mathsf{W}_{\xi'}$ and $|\xi'|\le|\xi|^{|\xi|}$ for all $\xi'\in \Xi$.
\end{lemma}

Lemma~\ref{theorembounded} implies that the number of VASSes in the sequence~(\ref{2023-07-24}) is bounded by $|\xi|^{|\xi|}$.
The VASSes in~(\ref{2023-07-24}) enjoy additional property.
We will come back to this in Section~\ref{sectionalmostnormal}.

\subsection{Reduction}\label{sectionreduction}

To make it easy for the following account, we consider runs in $\mathbb{N}^d_{\omega}$ by imposing the additional equality $\omega\pm n=\omega$ for every $n\in\mathbb{N}_{\omega}$.
For any $d$-dimensional $\xi_j=\mathbf{u}_jG_j\mathbf{v}_j$, if $p_j(\mathbf{u}_j)\xrightarrow{G_j}_{\mathbb{N}^d_{\omega}}p_j(\mathbf{u}'_j)$
and $q_j(\mathbf{v}'_j)\xrightarrow{G_j}_{\mathbb{N}^d_{\omega}}q_j(\mathbf{v}_j)$ for some $\mathbf{u}_j',\mathbf{v}_j'$ such that the following statements are valid, then $\mathbf{u}_jG_j\mathbf{v}_j$ is \emph{pumpable}; otherwise $\mathbf{u}_jG_j\mathbf{v}_j$ is \emph{nonpumpable}.
\begin{eqnarray}
&&\text{For each $k\in[d]$ with $\mathbf{u}_j(k)\ne\omega$, the strict inequality $\mathbf{u}_j'(k)> \mathbf{u}_j(k)$ holds.} \label{taiyuan-2023-07-18-a}\\
&&\text{For each $k\in[d]$ with $\mathbf{v}_j(k)\ne\omega$, the strict inequality $\mathbf{v}_j'(k)> \mathbf{v}_j(k)$ holds.} \label{taiyuan-2023-07-18-b}
\end{eqnarray}
If $\xi_j=\mathbf{u}_jG_j\mathbf{v}_j$ is pumpable for every $i$, $\xi$ is \emph{pumpable}; it is \emph{nonpumpable} otherwise.

We now show how to handle $\xi$ if it is not nonpumpable.
The idea is that if $\xi_j=\mathbf{u}_jG_j\mathbf{v}_j$ is nonpumpable, then there exists at least one dimension in which values in that dimension are bounded by $|\xi|$ throughout a walk inside $\xi_i$.
As pointed out in \cite{leroux2019reachability}, pumpability can be seen as a form of coverability, and the latter can be verified in $\textsf{exp}(|\xi|)$ time~\cite{rackoff1978covering}.
The following is from~\cite{leroux2019reachability}, which is essentially due to~\cite{rackoff1978covering}.
\begin{lemma}\label{lemmapump}
Suppose $V$ is a $d$-VASS, $\mathbf{u}_0\in \mathbb{N}_{\omega}^d$, and $c=|\{i|\mathbf{u}_0(i)\in\mathbb{N}\}|$.
If there is a run $p_0(\mathbf{u}_0)\xrightarrow{\mathbf{a}_1}_{\mathbb{N}_{\omega}^d}p_1(\mathbf{u}_1)\xrightarrow{\mathbf{a}_2}_{\mathbb{N}_{\omega}^d}\cdots\xrightarrow{\mathbf{a}_k}_{\mathbb{N}_{\omega}^d}p_k(\mathbf{u}_k)$
   such that for each $i\in [d]$ satisfying $\mathbf{u}_0(i)\in\mathbb{N}$ some $j\in[k]$ exists such that $\mathbf{u}_j(i)>C^{1+c^c}$, where $C$ is a natural number satisfying $C \geq |V|$.
   Then there exists a path $\pi$ such that $p_0(\mathbf{u}_0)\xrightarrow{\pi}p_0(\mathbf{u})$ satisfying $\mathbf{u}\geq (C\,{-}\,|V|){\cdot}\mathbf{1}$ and $|\pi|<C^{(c+1)^{c+1}}$.
\end{lemma}

Lemma~\ref{lemmapump} is extremely useful in that it allows one to carry out dimension reduction.
Suppose $\xi$ is nonpumpable.
By definition some $\xi_j=\mathbf{u}_jG_j\mathbf{v}_j$ is nonpumpable.
Without loss of generality, we may assume that $p_j(\mathbf{u}_j)$ is not pumpable in some $i$-th dimension in the sense that~(\ref{taiyuan-2023-07-18-a}) fails.
By Lemma~\ref{lemmapump} every run from $p_j(\mathbf{u}_j)$ to $q_j(\mathbf{v}_j)$ must fall in the region:
\[
\mathbb{B}_i \;\stackrel{\rm def}{=}\; \underset{i-1\ \mathrm{times}}{\underbrace{\mathbb{N}\times\ldots\times\mathbb{N}}} \times[0,B] \times \underset{d-i\ \mathrm{times}}{\underbrace{\mathbb{N}\times\ldots\times\mathbb{N}}}.
\]
By Lemma~\ref{lemmapump} we may set $B\defn (2|\xi|)^{1+d^d}$.
It now becomes clear how to {\em reduce} the dimension of a nonpumpable $\xi_j$.
Let $G_j^{-i}=\left(Q_j^{-i},T_j^{-i}\right)$, where $Q_j^{-i},T_j^{-i}$ are defined as follows:
\begin{eqnarray*}
Q_j^{-i} &=& \left\{(p,g) \mid p\in Q_j\ \mathrm{and}\ g\in[B]_{0}\right\}, \\
T_j^{-i} &=& \left\{(p,g)\stackrel{\mathbf{t}^{-i}}{\longrightarrow}(q,g\,{+\,}\mathbf{t}(i)) \mid
p\stackrel{\mathbf{t}}{\longrightarrow}q \in T_j,\ g,g\,{+\,}\mathbf{t}(i)\in[B]_{0}\right\}.
\end{eqnarray*}
In the above definition the notation $\mathbf{t}^{-i}$ stands for the vector obtained from $\mathbf{t}$ by removing the $i$-th entry, for example $(4,3,2,1)^{-2}=(4,2,1)$.
The construction of $\xi'=\mathbf{u}'_jG'_j\mathbf{v}'_j$ falls into one of the three categories:
\begin{itemize}
    \item If $\mathbf{u}_j(i),\mathbf{v}_j(i)\in\mathbb{N}$, then $\mathbf{u}'_j=\mathbf{u}_j,\mathbf{v}'_j=\mathbf{v}_j,G'_j=G^{-i}_j,p'_j=(p_j,\mathbf{u}_j(i)),q'_j=(q_j,\mathbf{v}_j(i))$.
    \item If $\mathbf{u}_j(i)\in\mathbb{N}$ and $\mathbf{v}_j(i)=\omega$, then $\mathbf{u}'_j=\mathbf{u}_j,G'_j=G^{-i}_j,p'_j=(p_j,\mathbf{u}_j(i)),q'_j=(q_j,r)$ for some $r\in[B]_0$, and $\mathbf{v}'_j$ differs from $\mathbf{v}_j$ only in that $\mathbf{v}'_j(i)=r$.
    \item If $\mathbf{v}_j(i)\in\mathbb{N}$ and $\mathbf{u}_j(i)=\omega$, then $\mathbf{v}'_j=\mathbf{v}_j,G'_j=G^{-i}_j,q'_j=(q_j,\mathbf{v}_j(i)),p'_j=(p_j,r)$ for some $r\in[B]_0$, and $\mathbf{u}'_j$ differs from $\mathbf{u}_j$ only in that $\mathbf{u}'_j(i)=r$.
\end{itemize}
In the latter two cases our reduction algorithm has to make a guess about $r$.
The nonpumpability of $\xi_j=\mathbf{u}_jG_j\mathbf{v}_j$ may also be caused by the failure of~(\ref{taiyuan-2023-07-18-b}).
This case can be treated like in the above case by reversing the transitions of $\xi_j=\mathbf{u}_jG_j\mathbf{v}_j$.

By analyzing the construction described in the above, one can prove the following lemma~\cite{leroux2019reachability}.
\begin{lemma}\label{lemmapumpdecomp}
Suppose $\xi$ is a $d$-dimensional standard nonpumpable KLM sequence.
A set $\Xi$ of KLM sequences of smaller rank can be computed from $\xi$ in $\mathsf{exp}(|\xi|^{d+d^d})$ time such that $\mathsf{W}_{\xi}=\bigcup_{\xi'\in\Xi}\mathsf{W}_{\xi'}$ and $|\xi'|\leq |\xi|^{d+d^d}$ for all $\xi'\in\Xi$.
\end{lemma}

The pumpability introduced in this section is slightly different from the one in~\cite{leroux2019reachability}.
We have not introduced the notion of rigidity.
A $d$-dimensional VASS $G$ may be reduced to a $(d{-}1)$-dimensional VASS $G^{-i}$ by encoding the values in the $i$-th entry into states.
If we still want to see $G^{-i}$ as $d$-dimensional, we impose the rigidity condition on the $i$-dimension of $G$.
This is significant because future tests of pumpability of $G$ will ignore all rigid dimensions.
In our algorithm it is never necessary to check the pumpability property of a VASS obtained by a reduction.
So rigidity is not necessary in our setting.

\section{Proof of the Main Result}
\label{sectionalmostnormal}

Section~\ref{sectionlinearpathscheme} and Section~\ref{sectionklmstalgo} have recalled the techniques applicable to the $d$-VASSes in general and the $2$-VASSes in particular, and have generalized the latter to the effectively $2$-dimensional $3$-VASSes.
In this section we use those techniques to construct an algorithm for $3$-VASS.
The basic idea is to apply the general KLMST constructions to a $3$-dimensional KLM sequence, and then replace every effectively $2$-dimensional VASSes in the sequence by its LPS system.
The idea works because the components produced by the decompositions defined in Section~\ref{sectiondimenreduct} are effectively $2$-dimensional and that being a linear path scheme is a property independent of its input/output locations.
The claim about the decomposition is proved in~\cite{leroux2019reachability}.
For completeness we provide a proof outline for the claim.
In the following proof and the rest of the section we shall assume that $\xi$ is a $3$-dimensional KLM sequence of the form~(\ref{2023-04-17}).
\begin{lemma}\label{decomp-2023-05-11}
Let $G_j'$ be the subgraph of $G_j$ consisting of all the unbounded edges of $G_j$.
If $G_{j'}\ne G_{j}$, then $V_{G_j'}\ne V_{G_j}$.
\end{lemma}
\begin{proof}
By our definition of~(\ref{2023-04-17}) the digraph $G_j$ is strongly connected.
Assume that $V_{G_j}=V_{G_j'}$.
Let $\Phi_j$ be the Parikh image of a cycle containing every edge of the graph $G_j$.
Then for some $k$ the Parikh image $k\Phi_j$ must be a linear combination $\sum_e\lambda_e\phi_e$ of a base of $V_{G_j'}$ with integer coefficients.
Let $\mathbf{h}^0$ be defined as in~(\ref{2023-07-21}).
Then $\mathbf{h}_j^0$ is the summation of all the solutions to the homogeneous equation system of $\xi_j$.
Let $k'$ be a large enough number such that $k'\mathbf{h}_j^0 -\Delta\left(\sum_e\lambda_e\phi_e\right)\ge\mathbf{0}$.
Now $k'\mathbf{h}_j^0=k\Phi_j+k'\mathbf{h}_j^0 -\Delta\left(\sum_e\lambda_e\phi_e\right)$.
So $k\Phi_j+k'\mathbf{h}_j^0-\sum_e\lambda_e\phi_e$ is a homogeneous solution, implying that every edge of $V_{G_j}$ is unbounded.
This is a contradiction.
So the assumption $V_{G_j}=V_{G_j'}$ must be invalid.
\end{proof}

\subsection{$3$-Normal KLM Sequences}
\label{subsectionwitnessofalmostnormalklm}

To provide a more informative characterization of $\xi$ that contains both $3$-dimensional VASSes and linear path schemes, we introduce the composite characteristic system $\mathscr{C}_{\xi,\Lambda}$ and its homogeneous version $\mathscr{C}^0_{\xi,\Lambda}$ indexed by an element $\Lambda\in\prod_{j\in J}\mathscr{L}_c(G_j)$, where $c$ is calculated from $\xi$ and $J\subseteq[n]_0$ is the set of the indexes $j$ such that $\mathbf{u}_jG_j\mathbf{v}_j$ is effectively $2$-dimensional.
Let the set of the remaining indexes be denoted by $\overline{J}=[n]_0\setminus J$.

\begin{definition}\label{CCS}
The {\em composite characteristic system} $\mathscr{C}_{\xi,\Lambda}$ for $\xi,\Lambda$ is defined as follows:
    \begin{eqnarray}
        \mathbf{y}_i&=&\mathscr{E}_{\Lambda_{i}}(\mathbf{x}_i), \text{ for }i\in J, \\
        \mathbf{y}_i&=&\mathbf{x}_i+\sum_{t=(p,\mathbf{a},q)}
        \phi_i(t){\cdot}\mathbf{a},\text{ for }i\in\overline{J}, \\
        \mathbf{1}_{q_i}-\mathbf{1}_{p_i}&=&\sum_{t=(p,\mathbf{a},q)\in T_i}
        \phi_i(t)(\mathbf{1}_{q}-\mathbf{1}_{p}),\text{ for }i\in \overline{J}, \\
        \mathbf{x}_{i+1}&=&\mathbf{y}_i+\mathbf{a}_{i+1},\text{ for }i\in[n]_0, \\
        \mathbf{x}_i &\sqsubseteq& \mathbf{u}_i,\text{ for }i\in[n]_0, \\
        \mathbf{y}_i &\sqsubseteq& \mathbf{v}_i,\text{ for }i\in[n]_0.
    \end{eqnarray}
\end{definition}
The {\em homogeneous composite characteristic system} $\mathscr{C}^0_{\xi,\Lambda}$ for $\xi,\Lambda$ is defined as follows:
    \begin{eqnarray}
        \mathbf{y}^0_i&=&\mathscr{E}^0_{\Lambda_{i}}(\mathbf{x}^0_i),\text{ for }i\in J \\
        \mathbf{y}^0_i&=&\mathbf{x}^0_i+\sum_{t=(p,\mathbf{a},q)}
        \phi^0_i(t){\cdot}\mathbf{a},\text{ for }i\in\overline{J}, \\
        \mathbf{0}&=&\sum_{t=(p,\mathbf{a},q)\in T_i}
        \phi^0_i(t)(\mathbf{1}_{q}-\mathbf{1}_{p}),\text{ for }i\in\overline{J}, \\
        \mathbf{x}^0_{i}[j]&=& 0,\text{ for }\mathbf{u}_i[j]\neq\omega,\text{ for }i\in[n]_0\text{ and }j\in[d], \\
        \mathbf{y}^0_{i}[j]&=& 0,\text{ for }\mathbf{v}_i[j]\neq\omega,\text{ for }i\in[n]_0\text{ and }j\in[d].
    \end{eqnarray}
For uniformity we keep the form of the variables of $\mathscr{C}_{\xi,\Lambda}$ as $V_{\xi}=\{(\mathbf{x}_i,\phi_i,\mathbf{y}_i)\}_{i\in[n]_0}$. If $\mathbf{u}_kG_k\mathbf{v}_k$ is effectively $2$-dimensional, $\phi_i$ are the variables in the LPS system $\mathbf{y}_i=\mathscr{E}_{\Lambda_{i}}(\mathbf{x}_i)$; otherwise $\phi_i$ is a Parikh image on $T_i$.
By definition $\mathscr{C}_{\xi,\emptyset}$ is the characteristic system defined in Section~\ref{sec-KLM-sequence-4-8}.
If $J=\{i_1,\ldots,i_k\} $ we also write $\Lambda$ as a sequence of the linear path schemes $\rho_{i_1}\ldots\rho_{i_k}$, where $\rho_{i_t}=\Lambda(i_t)$ for all $t\in[k]$.

For $j\in\overline{J}$, the component  $\xi_j=\mathbf{u}_jG_j\mathbf{v}_j$ is a $3$-dimensional VASS.
The definitions of being standard, unboundedness, and pumpability remain unchanged for $\xi_j$.
Recall that strong connectivity and pumpability are defined on the VASS $G_j$, and saturation and unboundedness are defined in terms of characteristic systems.
In the present situation we use the (homogeneous) composite characteristic system when defining the saturation and unboundedness properties.
Lemma~\ref{theroemstandard}, Lemma~\ref{theorembounded} and Lemma~\ref{lemmapumpdecomp} remain valid.

If $\xi_j=\mathbf{u}_jG_j\mathbf{v}_j$ is decomposed to $(\mathbf{u}_{j,1}G_{j,1}\mathbf{v}_{j,1})\ldots (\mathbf{u}_{j,n_j}G_{j,n_j}\mathbf{v}_{j,n_j})$, then by Lemma~\ref{decomp-2023-05-11} the VASS $G_{j,k}$ is effectively $2$-dimensional for all $k\in[n_j]$.
If $\xi_j=\mathbf{u}_jG_j\mathbf{v}_j$ is reduced to $\mathbf{u}_j'G_j'\mathbf{v}_j'$, then by construction $\mathbf{u}_j'G_j'\mathbf{v}_j'$ is $2$-dimensional.
If a $2$-VASS is seen as a special form of an effectively $2$-dimensional $3$-VASS, the following definition comes natural.

\begin{definition}\label{definitionalmostnormalklm}
The $3$-dimensional KLM sequence $\xi$ is \emph{$3$-normal} if, for all $k\in[n]_0$, either $\xi_k$ is standard, unbounded and pumpable or $\xi_k$ is effectively $2$-dimensional.
\end{definition}

The $3$-normal KLM sequences are good in the sense of the following theorem.
Its proof is routine.
\begin{theorem}\label{theoremwitnessalmostnormal}
If $\xi$ is $3$-normal, it has a witness bounded in size by $|\xi|^{O(|\xi|)}$.
\end{theorem}
\begin{proof}
Let $J=\{j_1,\ldots,j_k\}$ and $\Lambda=(\rho_{j_1},\ldots,\rho_{j_k})$
    be the tuple of linear path schemes.
    Consider the composite characteristic system $\mathscr{C}_{\xi,\Lambda}$.
    Let
    $\widehat{\mathbf{h}}=(\widehat{\mathbf{m}}_0,\widehat{\phi}_0,\widehat{\mathbf{n}}_0)\cdots(\widehat{\mathbf{m}}_n,
    \widehat{\phi}_n,\widehat{\mathbf{n}}_n)$ be a minimal solution to $\mathscr{C}_{\xi,\Lambda}$ and
    $\widehat{\mathbf{h}}^0=(\widehat{\mathbf{m}}^0_0,\widehat{\phi}^0_0,\widehat{\mathbf{n}}^0_0)\cdots(\widehat{\mathbf{m}}^0_n,
    \widehat{\phi}^0_n,\widehat{\mathbf{n}}^0_n)$ be a solution to the homogeneous system
    $\mathscr{C}^0_{\xi,\Lambda}$.
    By Lemma~\ref{2023-04-15} these systems render true the following statements.
    \begin{itemize}
       \item $|\widehat{\mathbf{h}}|\leq |\xi|^{O(|\xi|)}$ and $|\widehat{\mathbf{h}}_0|\leq |\xi|^{O(|\xi|)}$.
       \item For every $j\in\overline{J}$, the followings are valid.
       \begin{itemize}
       \item For every $t\in T_j$, $\widehat{\phi}^0_j(t)>0$.
       \item For every $g\in [3]$, $\mathbf{u}_j[g]=\omega$ implies $\widehat{\mathbf{m}}^0_j[g]>0$ and $\mathbf{v}_j[g]=\omega$
       implies $\widehat{\mathbf{n}}^0_j[g]>0$.
       \end{itemize}
    \end{itemize}
For every $j\in J$, the linear path scheme system $\mathcal{E}_{\rho_{j}}$ of $\rho_{j}$, formulated in Definition~\ref{LPSS}, is part of the composite characteristic system $\mathscr{C}_{\xi,\Lambda}$ for $\xi,\Lambda$, formulated in Definition~\ref{CCS}.
Let $\widetilde{\mathbf{h}}_j=(\widetilde{\mathbf{m}}_j,\widetilde{\phi}_j,\widetilde{\mathbf{n}}_j) \stackrel{\rm def}{=}(\widehat{\mathbf{m}}_j,\widehat{\phi}_j,\widehat{\mathbf{n}}_j)+r(\widehat{\mathbf{m}}^0_j,\widehat{\phi}^0_j,\widehat{\mathbf{n}}^0_j)$.
By definition $\widetilde{\mathbf{h}}_j$ is a solution to the linear path scheme system $\mathscr{E}_{\rho_j}$, that is $\widetilde{\mathbf{n}}_j=\mathscr{E}_{\rho_j}(\widetilde{\mathbf{m}}_j)$, where  $\rho_j=\alpha_{j,0}(\beta_{j,1})^*\alpha_{j,1}(\beta_{j,2})^*\cdots (\beta_{j,n_j})^*\alpha_{j,n_j}$.
The walk $\pi_j$ is of the following form
\begin{eqnarray}
p_j(\widetilde{\mathbf{m}}_j)\xrightarrow{\alpha_{j,0}}
\xrightarrow{(\beta_{j,1})^{\widetilde{\phi}_j(\beta_{j,1})}}\xrightarrow{\alpha_{j,1}}
\cdots\xrightarrow{(\beta_{j,n_j})^{\widetilde{\phi}_j(\beta_{j,n_j})}}\xrightarrow{\alpha_{j,n_j}}
q_j(\widetilde{\mathbf{n}}_j).
\end{eqnarray}
So it is guaranteed that every path defined by a solution to the composite characteristic system $\mathscr{C}_{\xi,\Lambda}$ for $\xi,\Lambda$ passes through $\rho_{j_1},\ldots,\rho_{j_k}$ with all locations in $\mathbb{N}^3$.

We only have to consider $\xi_k=\mathbf{u}_kG_k\mathbf{v}_k$ with $k\in\overline{J}$.
We hope to prove that there exists some $r\in\mathbb{N}$ that is not too large such that
\begin{eqnarray}\label{gigkj}
p_k(\widehat{\mathbf{m}}_k+r\widehat{\mathbf{m}}^0_k) &\xrightarrow{G_k}&
q_k(\widehat{\mathbf{n}}_k+r\widehat{\mathbf{n}}^0_k).
\end{eqnarray}
Since $\xi_k$ is pumpable, there is a circular walk $\psi_k$ from $p_k(\widehat{\mathbf{m}}_k)$ to some $p_k(\widehat{\mathbf{m}}')$ such that $\widehat{\mathbf{m}}'(i)>\widehat{\mathbf{m}}_k(i)$ for all $i\in[3]$ satisfying $\mathbf{u}_k(i)\ne\omega$.
Symmetrically there is a circular walk $\varphi_k$ from some $q_k(\widehat{\mathbf{n}}')$ to $q_k(\widehat{\mathbf{n}}_k)$ such that $\widehat{\mathbf{n}}'(i)>\widehat{\mathbf{n}}_k(i)$ for all $i\in[3]$ satisfying $\mathbf{v}_k(i)\ne\omega$.
Let $r_0$ be large enough such that $r_0\widehat{\phi}^0_k\ge\wp(\psi_k)+\wp(\varphi_k)$.
Let $\theta_k$ be the cycle defined by $r_0\widehat{\phi}^0_k-\wp(\psi_k)-\wp(\varphi_k)$, and let $\varpi_k$ be a path admitted by $\widehat{\phi}_k$.
What we have constructed is a path $\psi_k\varpi_k\theta_k\varphi_k$.
We can lift the unbounded entries of $\mathbf{u}_k,\mathbf{v}_k$ as large as necessary so that both $\psi_k$ and $\varphi_k$ are in the first octant.
But the path $\varpi_k$ and the cycle $\theta_k$ may go out of the first octant.
However we can repeat $\psi_k\theta_k\varphi_k$ for some $r_1$ times so that both $\varpi_k$ and $\theta_k$ stay completely in the first octant.
Let $r=r_0{\cdot}r_1$.
Then
\[
p_k(\widehat{\mathbf{m}}_k+r\widehat{\mathbf{m}}^0_k) \xrightarrow{(\psi_k)^r\varpi_k(\theta_k)^r(\varphi_k)^r}
q_k(\widehat{\mathbf{n}}_k+r\widehat{\mathbf{n}}^0_k).
\]
Let $\pi_k=(\psi_k)^r\varpi_k(\theta_k)^r(\varphi_k)^r$.
Since $(\theta_k)^r$ starts and ends in the first octant, we can choose $r_1$ so that $(\theta_k)^r$ stays completely in the first octant.
We have therefore proved~(\ref{gigkj}).
Clearly $|\pi_k|\leq |\widehat{\mathbf{h}}|+r|\widehat{\mathbf{h}}^0|$.

In summary there exists a walk $\pi=\pi_1\ldots\pi_n$ that satisfies
\begin{itemize}
\item $\Delta(\pi_j)=\sum_{t\in T_j}\widetilde{\phi}_j(t)$ for all $j\in[n]_0$, and
\item $p_j(\widetilde{\mathbf{m}}_j)\xrightarrow{\mathbf{a}_{j+1}}q_{j+1}(\widetilde{\mathbf{n}}_{j+1})$ for all $j\in[n\,{-}\,1]_0$.
\end{itemize}
By Lemma~\ref{pottier-lemma}, one may choose $r$ satisfying $r\leq |\xi|^{O(|\xi|)}$.
Thus $|\pi|<|\xi|^{O(|\xi|)}$.
\end{proof}

\subsection{The Algorithm $3$-\textsc{Klmst}}

We have shown that the $3$-normal KLM sequences have bounded witness.
In this section we propose a nondeterministic algorithm that can transform a $3$-dimensional KLM sequence $\xi$ to a $3$-normal KLM sequence such that $\xi$ has a witness if and only if a successful execution of the nondeterministic algorithm produces a $3$-normal KLM sequence.
The idea of the algorithm is simple.
Given a $3$-dimensional KLM sequence $\xi$ as in~(\ref{2023-04-17}), we apply the standardization, decomposition and reduction to $\xi$.
Whenever an effectively $2$-dimensional $\xi_k=\mathbf{u}_kG_k\mathbf{v}_k$ is constructed, the algorithm guesses a linear path scheme $\rho\in\mathscr{L}_c(G_k)$, where $c$ depends on $G_k$, and substitutes $\rho$ for $\xi_k$.
Let's continue to write $\xi_k$ for the $k$-th component after the replacement.
The key of our approach is that no further treatment to $\xi_k$ is then necessary.
The reason is that if $\mathbf{w}$ is a witness to $\xi$ then, for some choice of $\rho$, $\mathbf{w}$ is also a witness to the new KLM sequence.
So reachability is preserved in the sense of nondeterministic algorithm.

The formal definition of the algorithm, named $3$-\textsc{Klmst}, is given below, assuming that $\xi$ is the input of the form~(\ref{2023-04-17}).
\begin{enumerate}
\item If $\xi$ is not satisfiable, abort.
\item If $\xi$ is not standard, turn it nondeterministically to a standard KLM sequence $\xi'$.
\item For each $j'$, if $\xi_{j'}'$ is bounded, replace it nondeterministically by a decomposition of $\xi_{j'}'$.
\item For each $j'$, if $\xi_{j'}'$ is nonpumpable, replace it nondeterministically by a reduction of $\xi_{j'}'$.
\item Replace every effectively $2$-dimensional component $\xi_{j''}''$ of the new KLM sequence $\xi''$ nondeterministically by a linear path scheme.
\item If the new KLM sequence $\xi'''$ is $3$-normal, output $\xi'''$; otherwise apply $3$-\textsc{Klmst} to $\xi'''$.
\end{enumerate}
If the input $\xi$ has a witness $\mathbf{w}$, then by Lemma~\ref{theroemstandard}, Lemma~\ref{theorembounded} and Lemma~\ref{lemmapumpdecomp}, there exists a successful execution of $3$-\textsc{Klmst}$(\xi)$ that outputs a normal KLM sequence $\xi_{final}$ such that $\mathbf{w}$ is a witness of $\xi_{final}$.
Conversely suppose $\xi_{final}$ is the output of a successful execution of $3$-\textsc{Klmst}$(\xi)$.
By Theorem~\ref{theoremwitnessalmostnormal}, there is a witness $\overline{\mathbf{w}}$ to $\xi_{final}$.
By Lemma~\ref{theroemstandard}, Lemma~\ref{theorembounded} and Lemma~\ref{lemmapumpdecomp}, $\overline{\mathbf{w}}$ is also a witness to the input $\xi$.
Hence the correctness of $3$-\textsc{Klmst}.

The following lemma describes the complexity of $3$-\textsc{Klmst}.
\begin{lemma}\label{theoremcorrect3klmst}
Let $\xi=\mathbf{u}G\mathbf{v}$ be a $3$-dimenisonal KLM sequecne, and let $n$ be the number of transitions in $G$.
If $W_{\xi}\neq \emptyset$, then there exists a function $f(x)\stackrel{\rm def}{=}x^{x}$ such that if $\xi_{final}$ is the output of a successful execution of $3$-\textsc{Klmst}$(\xi)$, then $|\xi_{final}|\leq f^{2n}(|\xi|)$.
\end{lemma}
\begin{proof}
Let $\xi'=(\mathbf{u}_0G_0\mathbf{v}_0)\mathbf{a}_1(\mathbf{u}_1G_1
\mathbf{v}_1)\mathbf{a}_2\cdots \mathbf{a}_k(\mathbf{u}_kG_k
\mathbf{v}_k)$ be the KLM sequence obtained from $\xi$ by standardization.
Clearly $k\leq n-1$ and $G_i$ is strongly connected for all $i\in[k]_{0}$.
For each $i\in[k]_0$, if $\xi_i=\mathbf{u}_iG_i\mathbf{v}_i$ is bounded then the algorithm $3$-\textsc{Klmst} decomposes $\xi_i$ to some KLM sequence $(\mathbf{u}_{i,0}G_{i,0}\mathbf{v}_{i,0}) \mathbf{a}_{i,1}(\mathbf{u}_{i,1}G_{i,1}\mathbf{v}_{i,1})\mathbf{a}_{i,2} \cdots \mathbf{a}_{i,i_k}(\mathbf{u}_{i,i_k}G_{i,i_k}\mathbf{v}_{i,i_k})$.
By Lemma~\ref{decomp-2023-05-11}, all of $(\mathbf{u}_{i,0}G_{i,0}\mathbf{v}_{i,0}),(\mathbf{u}_{i,1}G_{i,1}\mathbf{v}_{i,1}), \cdots,(\mathbf{u}_{i,i_k}G_{i,i_k}\mathbf{v}_{i,i_k})$ are effectively $2$-dimensional.
For each $i\in[k]_0$, if $\xi_i=\mathbf{u}_iG_i\mathbf{v}_i$ is nonpumpable then the algorithm $3$-\textsc{Klmst} reduces $\xi_i$ to some $\mathbf{u}_i'G_i'\mathbf{v}_i'$.
By the definition of reduction, $\mathbf{u}_i'G_i'\mathbf{v}_i'$ is $2$-dimensional.
Since for all $i\in[k]_0$ the graph $G_i$ is strongly connected, further standardizations do not increase the number of the VASSes that have not been replaced by linear pathe schemes.
It follows that both the number of the standardizations and the number of the decompositions/reductions are bounded by $n$.
By Lemma~\ref{theroemstandard}, Lemma~\ref{theorembounded} and Lemma~\ref{lemmapumpdecomp}, every time a saturation operation or a decomposition/reduction operation is carried out, the size of the KLM sequence gets expanded by at most a ratio of $x^{x}$.
We conclude that $|\xi_{final}|\leq f^{2n}(|\xi|)$.
\end{proof}

The function $f^{2n}(x)$ is a tower function.
With this tower function in hand, we can design a simple-minded nondeterministic algorithm for $\mathbb{VASS}^3$ as follows: Upon receiving a $3$-dimensional KLM sequence $\xi$ of the form~(\ref{2023-04-17}), guess a $3$-dimensional KLM sequence $\xi'$ from $\mathbf{m}_0$ to $\mathbf{n}_0$ such that $\mathbf{m}_0\sqsubseteq\mathbf{u}_0$ and $\mathbf{n}_0\sqsubseteq\mathbf{v}_0$ and $|\xi'|<f^{2|\xi|}(|\xi|)$; accept if $\xi'$ is $3$-normal, reject otherwise.

We have proved Theorem~\ref{MAIN-THEOREM}.

\section{Conclusion}\label{sectionconclusion}

Our algorithm for the 3-VASS reachability problem applies the KLMST algorithm to convert an input KLM sequence nondeterministically to a KLM sequence so that every VASS in the output KLM sequence is either 3-normal or effectively $2$-dimensional.
The key property the algorithm relies on is that in the lower dimension, the length of a witness is bounded by a function on the size of graph and is independent of the first and the last locations.
The dimension reduction methodology ought to be instructive to the complexity theoretical study of the fixed dimension VASS reachability.

The best known lower bound for the $3$-VASS reachability problem is \textbf{PSPACE}-hard.
There is a huge gap between the currently known lower bound and upper bound.
It remains an open problem whether reachability in $3$-VASS is elementary.
But it seems unlikely that a more careful analysis of $3$-\textsc{Klmst} can obtain such a result.
Further research is necessary to narrow the gap.

\bibliography{vas}

\end{document}